%% file: main.tex
\documentclass[11pt]{article}

\usepackage[colorlinks=true,linkcolor=blue,citecolor=blue,urlcolor=blue]{hyperref}
\usepackage{xspace}
\usepackage{graphicx}
\usepackage{waingarten}
\usepackage[linesnumbered,ruled]{algorithm2e}
\usepackage{thmtools}
\usepackage{thm-restate}
\usepackage{color}
\usepackage{setspace}
\usepackage{array}

\def\FullBox{\hbox{\vrule width 8pt height 8pt depth 0pt}}
\newcommand{\QED}{\;\;\;\FullBox}
\renewenvironment{proof}{\noindent{{\textbf{Proof:}~}}} {\hfill\QED}
\providecommand{\email}[1]{\href{mailto:#1}{\nolinkurl{#1}\xspace}}

\def\FullBox{\hbox{\vrule width 8pt height 8pt depth 0pt}}

\makeatletter

\definecolor{b2}{RGB}{51,153,255}
\definecolor{mygreen}{RGB}{80,180,0}
\definecolor{brightmaroon}{rgb}{0.76, 0.13, 0.28}
\definecolor{lapislazuli}{rgb}{0.15, 0.38, 0.61}
\definecolor{indiagreen}{rgb}{0.07, 0.53, 0.03}

\newcommand{\ignore}[1]{}

\title{Polylogarithmic Sketches for Clustering}

\author {
Moses Charikar \and Erik Waingarten 
}

\begin{document}
\maketitle

\begin{abstract}
Given $n$ points in $\ell_p^d$, we consider the problem of partitioning points into $k$ clusters with associated centers.
The cost of a clustering is the sum of $p^{\text{th}}$ powers of distances of points to their cluster centers. 
For $p \in [1,2]$, we design sketches of size $\poly(\log(nd),k,1/\epsilon)$
such that the cost of the optimal clustering can be estimated to within factor $1+\epsilon$, 
despite the fact that the compressed representation does not contain enough information to recover the cluster centers or the partition into clusters.
This leads to a streaming algorithm for estimating the clustering cost with space $\poly(\log(nd),k,1/\epsilon)$.
We also obtain a distributed memory algorithm, where the $n$ points are arbitrarily partitioned amongst $m$ machines, each of which sends information to a central party who then computes an approximation of the clustering cost.
Prior to this work, no such streaming or distributed-memory algorithm was known with sublinear dependence on $d$ for $p \in [1,2)$.
\end{abstract}
\thispagestyle{empty}

\newpage

\tableofcontents

\thispagestyle{empty}

\newpage
\setcounter{page}{1}

\input{intro.tex}
\input{medians-sketch.tex}
\input{single-coordinate.tex}
\input{sampling-with-meta-data.tex}
\input{derandomization.tex}
\input{k-clustering.tex}

\appendix 

\section{Proof of (\ref{eq:ineq})}\label{app:importance-sampling-probs}

Here, we show that, for $y \in \R^d$ being the minimizer of $\sum_{i=1}^n \lambda_I \| y - x_i\|_p^p$ and $\sum_{i=1}^n \lambda_i x_i = 0$, the quantity $\sum_{i=1}^n \lambda_i |y_j - x_{ij}|^p$ is up to a constant factor, the same as $\sum_{i=1}^n \lambda_i |x_{ij}|^p$. Formally, for every $j \in [d]$,
\begin{align*}
\sum_{i=1}^n \lambda_i |y_j - x_{ij}|^p &\leq \sum_{i=1}^n \lambda_i |x_{ij}|^p \quad\text{(cost of $y_j$ smaller than $0$)}
\\&= \sum_{i=1}^n \lambda_i \left|x_{ij} - \sum_{i'=1}^n \lambda_{i'} x_{i'j}\right|^p \quad\text{($x_{1j}, \dots, x_{nj}$ are centered)}\\
&\leq \sum_{i=1}^n \lambda_i \left(\sum_{i'=1}^n \lambda_{i'} |x_{ij} - x_{i'j}|\right)^p \quad\text{(Triangle Inequality)} \\
&\leq \sum_{i=1}^n \sum_{i'=1}^n \lambda_i \lambda_{i'} |x_{ij} - x_{i'j}|^p, \quad\text{(Jensen's Inequality)}
\end{align*}
Furthermore,
\begin{align*}
\sum_{i=1}^n \sum_{i'=1}^n \lambda_i \lambda_{i'} |x_{ij} - x_{i'j}|^p &= \sum_{i=1}^n \sum_{i'=1}^n \lambda_i \lambda_{i'} |x_{ij} - y_{j} + y_j - x_{i'j}|^p \\
&\leq \sum_{i=1}^n \sum_{i'=1}^n \lambda_i \lambda_{i'} \left(|y_j - x_{ij}| + |y_j - x_{i'j}| \right)^p \quad\text{(Triangle Inequality)}\\
&\leq \sum_{i=1}^n \sum_{i'=1}^n 2^{p-1} \lambda_i \lambda_{i'} \left(|y_j - x_{ij}|^p + |y_j - x_{i'j}|^p \right) \quad\text{(H\"{o}lder Inequality)}\\
&= 2^p \sum_{i=1}^n \lambda_i |y_j - x_{ij}|^p.
\end{align*}

\input{medoids-sketch.tex}
\input{medoids-1-pass-lb.tex}

\pagenumbering{arabic}
\setcounter{page}{1}

\bibliographystyle{alpha}
\bibliography{waingarten,cluster}

\end{document}

%% file: intro.tex

\section{Introduction}

Given a large number of high-dimensional points, is it possible to compress the raw representation into a very compact sketch so that we can understand how clusterable the data is from just this highly compressed representation?
Given $n$ points in $d$ dimensions, we consider the problem of approximating the cost of clustering them into $k$ clusters from a compressed representation whose size is polylogarithmic in both $n$ and $d$.

For $n, d \in \N$, let $P = \{ x_1,\dots, x_n \} \in \R^d$ be any set of points with polynomially bounded entries (i.e., all coordinates may be represented with $O(\log(nd))$ bits). The $(k,z)$-clustering problem in $\ell_p^d$, asks to partition $P$ into at most $k$ clusters $C_1,\dots, C_k$ so as to minimize
\begin{align}
\sum_{\ell=1}^k \min_{c_{\ell} \in \R^d} \sum_{x \in C_{\ell}} \| x - c_{\ell} \|_p^z. \label{eq:k-z-clustering}
\end{align}
The problem is a generalization of the $k$-means and $k$-median problem; in particular, in Euclidean space $(p=2)$, $z = 2$ corresponds to $k$-means, and $z=1$ to $k$-median. 

We note that the raw representation of the dataset uses $O(nd \log(nd))$ bits, and that any algorithm which outputs optimal cluster centers $c_1,\dots, c_k \in \R^d$, or the optimal clustering $C_1,\dots, C_k$ must utilize $\Omega(kd)$, or  $\Omega(n \log k)$ bits of space, respectively. Hence, such algorithms cannot decrease the dependency on both $n$ and $d$ simultaneously. However, this does not rule out an exponential compression, from $O(nd \log(nd))$ bits to $\polylog(nd)$ bits (for constant $k$ and $z$), for algorithms which approximate the optimal clustering \emph{cost}, which only needs $O(\log(nd))$ bits. In this work, we show that it is possible to design sketches of size $\poly(\log(nd), 1/\eps)$ bits which $\eps$-approximate the optimal clustering \emph{cost}, despite the fact that we do not have enough information to compute the clusters nor the cluster centers which achieve such cost. 


Our results fit into a line of prior work on approximating the cost of optimization problems without necessarily computing an optimal solution. These have been investigated before for various problems and in various contexts, including estimating minimum spanning tree \cite{I04b,CEFMNRS05,CS09}, minimum cost matchings and Earth Mover's Distance \cite{I04b,AIK08,ABIW09,BI14}, minimum vertex cover and maximum matching \cite{PR07,NO08,YYI09,ORRR12,CKK20,KMNT20,B21} and model-fit \cite{KV18,KBV20,BGLT20}. Specifically for clustering problems, the value of the clustering cost is an important statistic; used, for example, in the ``elbow method'' for determining the number of clusters needed. We will show these sketches may be efficiently maintained on a stream as well as for distributed-memory models, implying $\polylog(nd)$-bit algorithms for these models of computation.


We start by reviewing a set of techniques in the literature to either reduce the dependence on the data set size or the dependence on the dimension.

\paragraph{Coresets.} The coreset technique is a ``dataset compression'' mechanism, aiming to reduce the dependency on $n$. From the $n$ points $P \subset \R^d$, an algorithm computes a much smaller (weighted) set of points $S \subset \R^d$, $w\colon S \to \R_{\geq 0}$, such that the cost of clustering the weighted points $S, w$ approximates that of $P$.  Following a long line of work~\cite{BHI02, HM04, AHV05, C09, LS10, FL11, FSS13, BFL16, SW18, HV20, CSS21}, the best coreset constructions for $(k,z)$-clustering in $\ell_p$ achieve sizes $\tilde{O}(k / \eps^4) \cdot \min\{ 1/\eps^{z-2}, k\}$ for a $(1\pm \eps)$-approximation. The ensuing streaming and distributed-memory algorithms maintain a coreset of the input; these algorithms find (approximately) optimal centers $c_1,\dots, c_k \in \R^d$ and use space complexity $d \cdot \tilde{O}(k/\eps^4) \cdot \min \{1/\eps^{z-2}, k\} \cdot \polylog(n)$.\footnote{The $\log n$-factors arise from utilizing the ``merge-and-reduce'' framework for maintaining coresets on a stream \cite{BS80, AHV05}, and the fact the coreset constructions are randomized.}

\paragraph{Dimension Reduction and Sketching.} In addition to constructing coresets, an algorithm may seek to optimize the dependence on $d$. 
There is a large body of work studying (oblivious) dimensionality reduction and sketching, where strong compression results are known for computing distances~\cite{AMS99, KOR00,SS02, BJKS04, CCF04, CM05, IW05, I06, JW09, KNW10, AKO11, AKR15, BBCKY17}. For example, for $p \in [1,2]$ there exists a (randomized) sketch $\sk \colon \R^d \to \R^t$ with $t$ much smaller than $d$ such that, for any two vectors $x, y \in \R^d$, an algorithm can approximate $\| x -y\|_p$ from $\sk(x)$ and $\sk(y)$ with high probability. While these results are encouraging, leveraging such sketches for distance computation in order to compress entire optimization problems like (\ref{eq:k-z-clustering}) is highly nontrivial. 
The challenge is that (\ref{eq:k-z-clustering}) implicitly considers distances among infinitely many vectors, and we need to rule out the possibility of spurious low cost solutions in the ``sketched'' space which do not have an analog in the original space. In particular, prior to this work, no streaming or distributed-memory algorithm was known which could reduce the dependence on $d$ for $p \in [1, 2)$.

There is one setting, of Euclidean space ($p = 2$), where one can sketch vectors while preserving (\ref{eq:k-z-clustering}). A sequence of works \cite{BZD10, CEMMP15, BBCGS19, MMR19} show that applying a Johnson-Lindenstrauss~\cite{JL84} map $\bPi \colon \R^d \to \R^t$ with $t = O(z^4 \log(k/\eps)/\eps^2)$, sketches Euclidean vectors to $O(t \log(nd))$ bits and preserves (\ref{eq:k-z-clustering}) up to $1\pm \eps$. We emphasize that Euclidean space $p = 2$ is special in this regard, because the Johnson-Lindenstrauss map achieves \emph{dimension reduction}, a property known not to hold in $\ell_1$~\cite{BC05, LN04, ACNN11}. In particular, $d$-dimensional vectors $x \in \R^d$ in Euclidean space are sketched to vectors $\bPi(x) \in \R^t$ in Euclidean space, i.e., one estimates $\| x - y\|_2$ by $\| \bPi(x) - \bPi(y)\|_2$. 
Thus the optimization problem (\ref{eq:k-z-clustering}) for $d$-dimensional Euclidean space reduces to the same optimization problem for a much smaller dimensional Euclidean space.
This can therefore be composed with known coreset constructions.
Importantly, the ``sketched'' space inherits all geometric properties of Euclidean spaces, a key aspect of prior works, and the reason they do not extend beyond Euclidean space. 
The technical challenge in applying sketches for $\ell_p$ when $p \neq 2$ is that the ``sketched'' space is non-geometric.\footnote{For example, the sketched space for $\ell_p$ with $p \neq 2$ does not satisfy the triangle inequality: it is not the case that for any $x, y, z \in \R^d$, the estimate of $(\sk(x), \sk(y))$ plus the estimate of $(\sk(y), \sk(z))$ is less than the estimate of $(\sk(x), \sk(z))$. On the other hand, for $p = 2$, the estimates of $(\sk(x), \sk(y))$, $(\sk(y), \sk(z))$, and $(\sk(x), \sk(z))$ are $\| \sk(x) - \sk(y)\|_2$, $\|\sk(y) - \sk(z)\|_2$, and $\|\sk(x) - \sk(z)\|_2$, so the triangle inequality does hold in the sketched space.}

\subsection{Our results} 

We give a streaming and distributed-memory algorithm for $(k,p)$-clustering in $\ell_p$ with space complexity $\poly(\log(nd), k, 1/\eps)$ bits. 

\begin{theorem}[Streaming $(k,p)$-Clustering in $\ell_p$]\label{thm:streaming}
For $p \in [1,2]$, there exists an insertion-only streaming algorithm which processes a set of $n$ points $x_1,\dots, x_n \in \R^d$ utilizing $\poly(\log(nd), k, 1/\eps)$ bits which outputs a parameter $\boldeta \in \R$ satisfying
\[ (1-\eps) \min_{\substack{C_1,\dots, C_k \\ \text{partition $[n]$}}} \sum_{\ell=1}^k \min_{c_{\ell} \in \R^d} \sum_{i \in C_{\ell}} \| x_i - c_{\ell} \|_p^p \leq \boldeta \leq (1+\eps) \min_{\substack{C_1,\dots, C_k \\ \text{partition $[n]$}}} \sum_{\ell=1}^k \min_{c_{\ell} \in \R^d} \sum_{i \in C_{\ell}} \| x_i - c_{\ell} \|_p^p \]
with probability at least $0.9$.
\end{theorem}

\begin{theorem}[Distributed-Memory $(k,p)$-Clustering in $\ell_p$]\label{thm:distributed}
For $p \in [1,2]$, there exists a public-coin protocol where $m$ machines receive an arbitrary partition of $n$ points $x_1,\dots, x_n \in \R^d$, each communicates $\poly(\log(md), k, 1/\eps)$ bits to a central authority who outputs a parameter $\boldeta \in \R$ satisfying 
\[ (1-\eps) \min_{\substack{C_1,\dots, C_k \\ \text{partition $[n]$}}} \sum_{\ell=1}^k \min_{c_{\ell} \in \R^d} \sum_{i \in C_{\ell}} \| x_i - c_{\ell} \|_p^p \leq \boldeta \leq (1+\eps) \min_{\substack{C_1,\dots, C_k \\ \text{partition $[n]$}}} \sum_{\ell=1}^k \min_{c_{\ell} \in \R^d} \sum_{i \in C_{\ell}} \| x_i - c_{\ell} \|_p^p \]
with probability at least $0.9$.
\end{theorem}

Both algorithms will follow from applying a coreset and compressing the representation of the coreset points into sketches to recover single-cluster cost. Specifically, the bottleneck for our algorithm will be estimating the cost of $(k,p)$-clustering in $\ell_p$ for $k = 1$. We give a linear sketch such that given a set of points $x_1,\dots, x_n \in \R^d$, one may approximate the $\ell_p^p$-median cost:
\[ \min_{y \in \R^d} \sum_{i=1}^n \| x_i - y\|_p^p. \]
Most of the technical work will be devoted to sketching this ``$\ell_p^p$-median cost'' objective. Then, the streaming and distributed-memory algorithm will evaluate the sum of $\ell_p^p$-median costs for all possible partitions of the coreset points into $k$ parts. The following theorem gives a linear sketch for approximating the $\ell_p^p$-median cost.\footnote{A related although different work is that of approximating the $\ell_p^p$-median (for instance, see Appendix F of \cite{BIRW16}). An $\ell_p^p$-median is a vector in $\R^d$ which means the sketch outputs $d$ numbers; however, we will sketch the $\ell_p^p$-median cost, which is a real number. Hence, our sketch will use $\poly(\log(nd), 1/\eps)$ space, as opposed to $\Omega(d)$ space needed to describe an $\ell_p^p$-median.} 

\begin{theorem}[$\ell_p^p$-Median Sketch]\label{thm:ell-p-median}
For $p \in [1,2]$, there exists a linear sketch which processes a set of $n$ points $x_1,\dots, x_n \in \R^d$ into a vector $\R^t$ with $t = \poly(\log(nd), 1/\eps)$ and outputs a parameter $\boldeta \in \R$ satsifying
\[ (1-\eps) \min_{y \in \R^d} \sum_{i=1}^n \| x_i - y\|_p^p \leq \boldeta \leq (1+\eps) \min_{y \in \R^d} \sum_{i=1}^n \|x_i - y\|_p^p \]
with probability at least $0.9$.
\end{theorem}

There are a few important remarks to make:
\begin{itemize}
\item The requirement that $p \leq 2$ is necessary for the exponential compression we desire. For $p > 2$, there are strong lower bounds for sketching distances which show that such sketches require $\Omega(d^{1-2/p})$ space~\cite{BJKS04}. For $p < 1$, we are not aware of small coresets.
\item The focus of this work is on optimizing the space complexity of the sketch, and while we do not explicitly specify the running time of the sketching and streaming algorithms, a naive implementation runs in time $(k\log(nd) / \eps)^{(k\log n/\eps)^{O(1)}}$. The exponential factor is due to the fact that we evaluate the cost of all possible partitions of the $(k\log(n)/\eps)^{O(1)}$-coreset points into $k$ clusters. One could alleviate the exponential dependence to $(k/\eps)^{O(1)}$ (as opposed to $(k \log n/\eps)^{O(1)}$) by running more sophisticated approximation algorithms \cite{BHI02, KSS04} on the sketched representation of the coreset.\footnote{The one subtlety is that the algorithm should be implemented without explicitly considering the $d$-dimensional representation of the points. Instead, it should only use the sketches of Theorem~\ref{thm:ell-p-median}.} We note that a super-polynomial dependence on $k$ should is unavoidable, because $(1\pm\eps)$-approximations for $(k,z)$-clustering problems, for \emph{non-constant} $k$, are NP-hard \cite{ACKS15, LMW17, CK19}. 
\item It would be interesting to generalize Theorem~\ref{thm:streaming} to dynamic streams. The reason our algorithm works in the insertion-only model is that we utilize the coreset of \cite{HV20} with the merge-and-reduce framework \cite{BS80, AHV05} which do not support deletions. While there exist dynamic coreset constructions~for the streaming model~\cite{BFLSY17, HSFZ19}, our use of coresets is not entirely black-box. Other dynamic coresets, like \cite{HK20}, focus on update time and do not optimize the space complexity. We must ensure that the algorithm for constructing the coreset does not utilize the $d$-dimensional representation of the dataset points. The coreset construction of \cite{HV20} only consider distances between the dataset points, so it suffices for us to only maintain a sketch of the dataset points. 
\item The fact that $z = p$ in our theorems above is a consequence of our techniques. It is unclear to us whether this assumption is necessary, although our approach hinges on the fact $\ell_p^p$ is additive over the $d$ coordinates. We leave this as a problem for future work.
\end{itemize}

A similar, yet importantly different notion of $(k,z)$-clustering considers \emph{medoid} cost, where the centers of the $k$ clusters $c_1,\dots, c_k$ are restricted to be dataset points. While seemingly similar to the $(k,z)$-clustering objective where centers are unrestricted, these two are qualitatively very different from a sketching perspective. In Appendix~\ref{app:medoid-cost}, we show that while a two-pass sketching algorithm may $\eps$-approximate the medoid cost, $\eps$-approximations for one-pass sketching algorithm require polynomial space.

\input{overview.tex}

%% file: overview.tex

\subsection{Technical Overview}

We give an overview of Theorem~\ref{thm:ell-p-median}. Once that is established, combining the $\ell_p^p$-median sketch with coresets, thereby establishing Theorems~\ref{thm:streaming} and~\ref{thm:distributed} is (relatively) straight-forward (for more details, see Section~\ref{sec:streaming}). Recall that for $p \in [1, 2]$, we will process $n$ points $x_1,\dots, x_n \in \R^d$, and aim to return an approximation to the $\ell_p^p$-median cost:
\begin{align}
\min_{c \in \R^d} \sum_{i=1}^n \| x_i - c\|_p^p. \label{eq:ell-p-median}
\end{align}
It will be useful to assume that the points are centered, i.e., $\sum_{i=1}^n x_i = 0 \in \R^d$ (we can enforce this because our sketches will be linear). 
The approach will come from the fact that the above optimization problem decomposes into a sum of $d$ independent optimizations, one for each coordinate, and (\ref{eq:ell-p-median}) seeks to evaluate the sum. Specifically, we may write
\begin{align*}
\min_{c \in \R^d} \sum_{i=1}^n \| x_i - c \|_p^p = \sum_{j=1}^d \min_{c_{j} \in \R} \sum_{i=1}^n |x_{ij} - c_j|^p.  
\end{align*}
Furthermore, for any fixed $j \in [d]$, estimating 
\begin{align}
\min_{c_j \in \R} \sum_{i=1}^n |x_{ij} - c_j|^p \label{eq:one-coord}
\end{align}
is much more amenable to $\ell_p$-sketching. Specifically, we let $x_{\cdot, j} \in \R^n$ be the vector containing the $j$-th coordinates of all $n$ points, and $\ind \in \R^n$ be the all-1's vector. Then, the quantity $\sum_{i=1}^n |x_{ij} - c_j|^p = \| x_{\cdot, j} - c_j \ind \|_p^p$, and since the $\ell_p$-sketches are linear, an algorithm may maintain $\sk(x_{\cdot, j}) \in \R^t$ (for $t = \poly(\log(nd))$) and after processing, could iterate through various values of $c_j \in \R$ to evaluate
\[ \left(\sum_{i=1}^n |x_{ij} - c_j|^p\right)^{1/p} = \|x_{\cdot, j} - c_j \ind \|_p \approx_{1\pm \eps} \text{ estimate of }(\sk(x_{\cdot, j}), \sk(c_j \ind)), \] 
and output the smallest value of $c_j \in \R$ found. In order to guarantee an $(1\pm \eps)$-approximation of (\ref{eq:one-coord}), only $\poly(1/\eps)$ values of $c_j$ need to be tried (since first evaluating $c_j = 0$ will specify the range where the optimal $c_j$ may lie). A simple union bound implies that for any fixed $j \in [d]$, we can prepare a small sketch $\sk(x_{\cdot, j})$ from which we can approximate (\ref{eq:one-coord}).

In summary, we want to estimate the sum of $d$ minimization problems. Even though each of the $d$ problems could be solved independently with a linear sketch, we do not want to process $d$ linear sketches (as this increases space). In addition, we do not know which of the $d$ minimizations will significantly affect the sum; hence, if we only (uniformly) sampled few $\bj_1, \dots, \bj_t \sim [d]$ and only processed $t \ll d$ sketches along the sampled dimensions, the variance of the estimator may be too large, making it completely useless. The technique we will use was recently developed in \cite{CJLW22}, building on \cite{AKO10, JW18}, under the name ``$\ell_p$-sampling with meta-data.'' In this paper, we further develop the ideas, and apply them to sketches for clustering in a simple and modular way. We refer the reader to Remark~\ref{rem:comparison} (following this technical overview), where we expand on the comparison to \cite{CJLW22}.

The goal is to approximate the sum of the $d$ minimization problems by \emph{importance sampling} (see Chapter~9 of \cite{O13}). While importance sampling is a well-known technique, it's use in (one-pass) linear sketching algorithms is counter-intuitive, and we are not aware of any linear sketches which use importance sampling in the literature, expect for this and the recent work of~\cite{CJLW22, CJKVY22}. Importance sampling will aim to estimate (\ref{eq:ell-p-median}) by sampling with respect to an alternate distribution $\calD$. In particular, (\ref{eq:ell-p-median}) may be re-written as
\begin{align}
d \cdot \Ex_{\bj \sim [d]}\left[ \min_{c_{\bj} \in \R} \sum_{i=1}^n |x_{i\bj} - c_{\bj}|^p \right] = d \Ex_{\bj \sim \calD}\left[ \bY_{\bj} \right] \qquad \text{where}\qquad \bY_{\bj} \eqdef \min_{c_{\bj} \in \R} \sum_{i=1}^n |x_{i\bj} - c_{\bj}|^p \cdot \dfrac{1}{\Pr_{\calD}[\bj]}, \label{eq:importance-sampling}
\end{align}
where $\calD$ is a distribution chosen so the variance of the random variable $\bY_{\bj}$ for $\bj \sim \calD$ is bounded. Once the variance of the random variable is bounded, only a few samples are needed to estimate its expectation in (\ref{eq:importance-sampling}). In general, the alternate distribution $\calD$ depends on the data in order to decrease the variance; for instance, coordinates $j \in [d]$ whose value of (\ref{eq:one-coord}) is higher should be sampled more often. Hence, importance sampling inherently interacts with the data in a two-stage process: 1) first, it samples $\bj \sim \calD$ (where the distribution is data-dependent), and 2) second, it evaluates $\bY_{\bj}$ by using (\ref{eq:one-coord}) and $\Prx_{\calD}[\bj]$ for the value $\bj \in [d]$ specified in the first step. 

In a two-pass algorithm, the two steps may be implemented sequentially. A \emph{sampling} sketch, like that of  \cite{JW18}, is used to sample $\bj \sim \calD$ in the first pass. In the second pass, the algorithm knows the value of the sampled $\bj$, so it maintains a sketch $\sk(x_{\cdot, \bj})$ of size $t$ and a sketch $\sk'(\Prx_{\calD}[\bj])$ of size $t'$ (to estimate $\Pr_{\calD}[\bj]$) from which it can evaluate the random variable $\bY_{\bj}$. The counter-intuitive aspect is that, in this case, we will perform both steps in one-pass: 
\begin{itemize}
\item We will use an $\ell_p$-sampling sketch of \cite{JW18} to sample from an importance sampling distribution $\calD$, and
\item Concurrently, we prepare $2d$ linear sketches: $d$ sketches $\sk(x_{\cdot, j})$ to evaluate (\ref{eq:one-coord}), one for each $j \in [d]$, and $d$ sketches $\sk'(\Pr_{\calD}[j])$ to evaluate $\Pr_{\calD}[j]$, one for each $j \in [d]$. The non-trivial part is to \emph{sketch the sketches}. by compressing the $2d$ linear sketches into a $O(\polylog(nd))$-bit Count-Min data structure \cite{CM05}.
\end{itemize}
The guarantee will be that the $\ell_p$-sampling sketch of \cite{JW18} generates a sample $\bj \sim \calD$, and the Count-Min data structure can recover an approximation 
\[ \widehat{\sk_1} \approx \sk(x_{\cdot, \bj}) \qquad \text{and}\qquad \widehat{\sk_2} \approx \sk'(\Pr_{\calD}[\bj]). \]
Furthermore, the sketch evaluation algorithm, which executes on the approximation $\widehat{\sk_1}$ and $\widehat{\sk_2}$, should be able to recover $(1\pm \eps)$-approximations to (\ref{eq:one-coord}) and $\Pr_{\calD}[\bj]$, so that the ratio of the two is a $(1\pm 2\eps)$-approximation to $\bY_{\bj}$. 

While the above plan provides a general recipe for importance sampling, the idea of ``sketching the sketches'' may not be applied in a black-box manner. First, the alternate distribution $\calD$ should admit a sampling sketch. Second, the sketch evaluation algorithm for $\sk(x_{\cdot, j})$ and $\sk'(\Pr_{\calD}[j])$ should be robust to the errors introduced by the Count-Min compression. Bounding the errors introduced by the Count-Min data structure, and ensuring that the approximate sketches $\widehat{\sk_1}$ and $\widehat{\sk_2}$ constitutes the bulk of the technical work. Specifically for us, the plan is executed as follows: when $\sum_{i=1}^n x_i = 0 \in \R^d$, every $j$ satisfies (see Appendix~\ref{app:importance-sampling-probs})
\begin{align}
\dfrac{\min_{c_j \in \R} \sum_{i=1}^n |x_{ij} - c_j|^p}{\| x_{\cdot, j}\|_p^p} \in [2^{-p}, 1]. \label{eq:to-estimate}
\end{align}
Hence, we will let $\calD$ be the distribution supported on $[d]$ given by setting, for each $j \in [d]$,
\begin{align*}
\Prx_{\bj \sim \calD}\left[\bj = j \right] &= \frac{\| x_{\cdot,j}\|_p^p}{Z} \qquad \text{where} \qquad Z = \sum_{j=1}^d \| x_{\cdot, j}\|_p^p = \sum_{i=1}^n \sum_{j=1}^d |x_{ij}|^p.
\end{align*}
Note that (\ref{eq:to-estimate}) implies the variance of $\bY_{\bj}$ for $\bj \sim \calD$ is appropriately bounded. Furthermore, since $\calD$ is an $\ell_p$-sampling distribution, the $\ell_p$-sampling sketches of \cite{JW18} are useful for sampling $\bj \sim \calD$. Finally, the approach of \cite{JW18} is particularly suited for bounding the errors incurred by Count-Min on $\widehat{\sk_1}$ and $\widehat{\sk_2}$, which we overview below.

At a high level, the $\ell_p$-sampling sketch of \cite{JW18} generates a sample $\bj$ from $[d]$ by identifying a \emph{heavy hitter} in a random scaling of the vector specifying the sampling probabilities. In particular, the algorithm generates $\bu_1,\dots, \bu_d \sim \mathrm{Exp}(1)$ and identifies an entry $j \in [d]$ in the vector
\[ \left( \frac{\|x_{\cdot, 1}\|_p}{\bu_{1}^{1/p}}, \frac{\|x_{\cdot, 2}\|_p}{\bu_{2}^{1/p}}, \dots, \frac{\|x_{\cdot, d-1}\|_p}{\bu_{d-1}^{1/p}},  \frac{\|x_{\cdot, d}\|_p}{\bu_{d}^{1/p}}\right) \in \R^d,\]
whose value satisfies 
\begin{align}
\dfrac{\|x_{\cdot, j}\|_p}{\bu_{j}^{1/p}} \gsim \left(\sum_{j'=1}^d \frac{\|x_{\cdot, j'}\|_p^p}{\bu_{j'}} \right)^{1/p}, \label{eq:entry-lb}
\end{align}
and is the largest among those heavy hitters. For the coordinate $\bj \in [d]$ recovered by the $\ell_p$-sampling sketch \cite{JW18}, the inequality (\ref{eq:entry-lb}) gives a lower bound on how large $1/\bu_{\bj}^{1/p}$ will be. In particular, by applying the same transformation to the vector of \emph{sketches},
\begin{align} 
\left(\dfrac{\sk(x_{\cdot, 1})}{\bu_1^{1/p}},\dots, \dfrac{\sk(x_{\cdot, d})}{\bu_d^{1/p}} \right) \in (\R^{t})^d \qquad\text{and}\qquad \left(\dfrac{\sk'(\Pr_{\calD}[1])}{\bu_1^{1/p}}, \dots, \dfrac{\sk'(\Pr_{\calD}[d])}{\bu_{d}^{1/p}} \right) \in (\R^{t'})^d,  \label{eq:sketched-vectors}
\end{align}
the $t$ and $t'$ coordinates corresponding to the sketches $\sk(x_{\cdot, \bj}) \in \R^t$ and $\sk'(\Pr_{\calD}[\bj]) \in \R^{t'}$ will be heavy hitters of those vectors as well. \ignore{Namely, 
\newcommand{\berr}{\boldsymbol{\mathrm{err}}}
\begin{align*}
\forall h \in [t] \qquad&: \qquad \left| \dfrac{\sk(x_{\cdot, \bj})_h}{\bu_{\bj}^{1/p}} \right| \gsim \left( \sum_{j'=1}^d \dfrac{\sk(x_{\cdot, j'})_h^p}{\bu_{j'}}\right)^{1/p} \eqdef \berr_{1, h},\text{ and } \\
\forall h \in [t'] \qquad&:\qquad \left|\dfrac{\sk(\Pr_{\calD}[\bj])_h}{\bu_{\bj}^{1/p}} \right| \gsim \left( \sum_{j'=1}^d \dfrac{\sk(\Pr_{\calD}[j'])_h^p}{\bu_{j'}}\right)^{1/p} \eqdef \berr_{2,h}. 
\end{align*}}
Namely, with only $\poly(\log(nd), 1/\eps)$-bits, the Count-Min data structure will recover the entries of $\sk(x_{\cdot, \bj})$ and $\sk'(\Pr_{\calD}[\bj])$ up to a small additive error, proportional to the $\ell_1$-norm of (\ref{eq:sketched-vectors}). We know the distribution of sketched vectors (\ref{eq:sketched-vectors}) (since these are simply $\ell_p$-sketches~\cite{I06}), so we will be able to bound the additive error and show that the sketch evaluation algorithms of $\widehat{\sk_1}$ and $\widehat{\sk_2}$ return the desired $(1\pm \eps)$-approximations.
\ignore{

We will use the Count-Min~\cite{CM05} sketch to compressing the $d$ linear sketches (the vectors $\sk(x_{\cdot, j})$ for all $j$ and $\sk() into $\poly(\log(nd), 1/\eps)$ bits. This means that the we can recover the entries of the $d$ linear sketches up to an additive error, which depends on the sum of all entries 

Technically, we will not be able to recover the sketch $\sk(x_{\cdot, \bj})$ exactly. The reason is that 

There is one more challenge: a sample $\bk \sim \calD$ that we produce, and the sketches to evaluate the numerator and denominator \emph{for that sample} are both data-dependent. Technically, we will not be able to obtain an error-less  sketch for evaluating the quantity in the expectation (because we needed to have prepared for potentially other $k' \in [d]$ to have been sampled), and these will incur an additive error. As we will see, the additive error produced will be controlled via a transformation of the input akin to the ``precision sampling'' framework of \cite{AKO10, JW18}. This technique of using precision sampling to sample entire sketches was recently used in \cite{CJLW22} under the name ``$\ell_p$-sampling with meta-data.'' Here, the meta-data will correspond to the sketch for the numerator and the denominator.}

\begin{remark}[Comparison to \cite{CJLW22}]\label{rem:comparison}
The technique, ``$\ell_p$-sampling with meta-data'', arises in \cite{CJLW22} in the following context. They seek a linear sketch $\sk \colon \R^d \to \R^t$ which can process a vector $y \in \R^d$ and evaluate a weighted $\ell_1$-norm, $\sum_{i=1}^d w_i(y) \cdot |y_i|$, where the weights $w_1(y), \dots, w_d(y) \in \R_{\geq 0}$ are themselves dependent on the vector $y$. This arises as an algorithmic step in streaming algorithms for geometric minimum spanning tree and the earth-mover's distance. Mapping the above formulation to our setting, we want to evaluate a weighted $\ell_1$-norm as well, where the $i$-th weight corresponds to $\Prx_{\bj \sim \calD}[\bj = i]$, and the $i$-th value seek to sum is $\bY_i$ (as in (\ref{eq:importance-sampling})). The perspective of this technique as importance sampling (as presented in this work) is new. Indeed, the appropriate setting of weights is only apparent once one multiplies and divides the contribution of the $j$-th coordinate by $\| x_{\cdot, j}\|_p^p$ to define $\calD$.
\end{remark}

%% file: medians-sketch.tex

\section{Sketching Median Costs}

\subsection{Statement of Main Lemma} 

\begin{theorem}\label{thm:median-sketch}
Fix $n, d \in \N$, as well as $p \in [1,2]$ and $\eps, \delta \in (0,1)$. There exists a linear sketch using $\poly(\log d, 1/\eps, \log(1/\delta))$ space which processes a set of $n$ points $x_1, \dots, x_n \in \R^d$, and outputs a parameter $\boldeta \in \R$ which satisfies
\[ \min_{y \in \R^d} \sum_{i=1}^n \| y - x_i\|_p^p \leq \boldeta \leq (1+\eps) \min_{y \in \R^d} \sum_{i=1}^n \| y - x_i \|_p^p\]
with probability at least $1 - \delta$. 
\end{theorem}

We work with the following representation of a linear sketch. The processed set of $n$ points in $\R^d$ are stacked to form a vector $x \in \R^{nd}$. A linear sketch using space $s$ is a distribution $\calM$ supported on $s \times (nd)$ matrices. The theorem states that for any fixed $x_1, \dots, x_n \in \R^d$, with probability $1-\delta$ over the draw of $\bS \sim \calM$, an algorithm with access to the vector $\bS x \in \R^s$ and $\bS$ can output $\boldeta$ satisfying the above guarantees.

Linear sketches of the above form imply efficient streaming algorithms, albeit with some subtleties. It is useful to first assume that the streaming algorithm can store its randomness for free (we will address this in Subsection~\ref{sec:randomness}) so that it knows the matrix $\bS$. In particular, since $\bS \in \R^{s \times nd}$ acts on the vector $x \in \R^{nd}$ which vertically stacks $x_1,\dots, x_n \in \R^d$, the columns of $\bS$ may be broken up into $n$ groups of size $d$, so
\[ \bS = \left[\begin{array}{cccc} \bS_1 & \bS_2 & \dots & \bS_n \end{array} \right], \qquad \text{and} \qquad \bS x = \sum_{i=1}^n \bS_i x_i. \] In the \emph{insertion-only} model, an algorithm would process the points one-at-a-time, and at time-step $j$, maintain $\sum_{i=1}^j \bS_i x_i \in \R^{s}$. In the \emph{turnstile} model of streaming, there is a subtlety in the implementation; namely, as the algorithm receives insertions and deletions of points in $\R^d$, it must know which index $i \in [n]$ it is considering. The reason is that the algorithm should know which of the sub-matrix $\bS_i$ to update the point with. 

For our application of the $\ell_p^p$-median sketch to $(k,p)$-clustering in $\ell_p$, we consider a weighted $\ell_p^p$-median. Namely, for points $x_1,\dots, x_n \in \R^d$ and weights $\lambda_1,\dots, \lambda_n \in [0,1]$ with $\sum_{i=1}^n \lambda_i = 1$, the $\ell_p^p$-median cost with respect to weights $\lambda_1,\dots, \lambda_n$ is
\[ \min_{y \in \R^d} \sum_{i=1}^n \lambda_i \| y - x_i \|_p^p. \]
It is useful to first consider of $\lambda_1 = \dots = \lambda_n = 1/n$. For general weights, the sketch will receive as input $\bS = [\bS_1,\dots, \bS_n] \in \R^{s \times (nd)}$, the vector $\sum_{i=1}^n \lambda_i^{1/p} \bS_i x_i \in \R^s$, and the weights $\lambda_1,\dots, \lambda_n$.

\paragraph{Centering Points} There is a straight-forward way to process the points so as to assume they are centered. Specifically, the average point may be subtracted from every point by applying a linear map, and since our sketch is linear, subtracting the average point may be incorporated into the sketch. For weights $\lambda_1,\dots, \lambda_n \in [0,1]$ satisfying $\sum_{i=1}^n \lambda_i = 1$, we consider the linear map
\[ (x_1, \dots, x_n) \mathop{\mapsto} \left(x_1 - \sum_{i=1}^n \lambda_i x_i, \dots, x_n - \sum_{i=1}^n \lambda_i x_i\right) \in \R^{nd}.  \]
Hence, we assume, without loss of generality, that the points $x_1, \dots, x_n \in \R^d$ satisfy
\begin{align}
\sum_{i=1}^n \lambda_i x_i = 0 \in \R^d. \label{eq:centered}
\end{align}
\ignore{Equivalently, the algorithm receives access to the sketched vector
\[ \sum_{i=1}^n \lambda_i^{1/p}\bS_i \left(x_i - \sum_{j=1}^n \lambda_j x_j \right) = \sum_{i=1}^n \lambda_i^{1/p} \bS_i x_i -  \sum_{i=1}^n \sum_{j=1}^n \lambda_i^{1/p} \lambda_{j} \bS_i x_{j} = \sum_{i=1}^n \lambda_i^{1/p} \bS_i - \sum_{j=1}^n \lambda_j^{1/p}.  \]}

The centering is useful for deriving the following set of inequalities, which will be useful for our sketching procedures. Suppose we denote $y \in \R^d$ as the point which minimizes $\sum_{i=1}^n \lambda_i \| y - x_i\|_p^p$. Then, for every $j \in [d]$, 
\begin{align*} 
\sum_{i=1}^n \lambda_i | y_j - x_{ij} |^p \leq \sum_{i=1}^n \lambda_i | x_{ij} |^p \leq 2^p \sum_{i=1}^n \lambda_i |y_j - x_{ij}|^p. 
\end{align*}
Importantly for us, every $j \in [d]$ satisfies
\begin{align}
2^{-p} \leq \dfrac{\min_{y_j \in \R} \sum_{i=1}^n \lambda_i |y_j - x_{ij}|^p}{\sum_{i=1}^n \lambda_i |x_{ij}|^p} \leq 1. \label{eq:ineq}
\end{align}

We let $\calD$ be the distribution supported on $[d]$ given by setting, for each $j \in [d]$,
\begin{align*}
\Prx_{\bj \sim \calD}\left[\bj = j \right] &= \frac{1}{Z} \sum_{i=1}^n \lambda_i |x_{ij}|^p \qquad \text{where} \qquad Z = \sum_{j' = 1}^d \sum_{i=1}^n \lambda_i |x_{ij'}|^p = \sum_{i=1}^n \lambda_i \|x_i\|_p^p. 
\end{align*}
Then, the quantity we want to estimate may be equivalently re-written as:
\begin{align}
\sum_{j=1}^d \min_{y_j \in \R} \sum_{i=1}^n \lambda_i |y_j - x_{ij}|^p &= Z \cdot \Ex_{\bj \sim \calD}\left[ \frac{\min_{y_{\bj} \in \R} \sum_{i=1}^n \lambda_i |y_{\bj} - x_{i\bj}|^p}{\sum_{i=1}^n \lambda_i |x_{i\bj}|^p} \right],\label{eq:quantity-to-estimate}
\end{align}
where the value within the expectation is bounded between $2^{-p}$ and $1$. Furthermore, the quantity $Z$ will be sketched with an $\ell_p$-sketch, and a sample $\bj \sim \calD$ will be drawn with an $\ell_p$-sampling sketch. Hence, the plan is to produce $t = O(1 / \eps^2)$ samples of $\bj_1, \dots, \bj_t \sim \calD$, and produce a sketch to evaluate the numerator inside the expectation, and the denominator inside the expectation. Taking an empirical average of the samples to estimate the expectation, and multiplying it by the estimate of $Z$ will give the desired estimator. 

\begin{lemma}[Main Lemma] \label{lem:main-sketch-lemma}
For any $n, d \in \N$, $p \in [1,2]$ and $\eps, \delta \in (0, 1)$, let $s = \poly(\log d, 1/\eps, 1/\delta)$.\footnote{See (\ref{eq:space-complexity}) for the specific polynomial bounds.} There exists a distribution $\calS$ over $s \times (nd)$ matrices, and an algorithm such that for any $n$ vectors $x_1,\dots, x_n \in \R^d$ and any $\lambda_1,\dots, \lambda_n \in [0,1]$ with $\sum_{i=1}^n \lambda_i = 1$ and $\sum_{i=1}^n \lambda_i x_i = 0$, the following occurs:
\begin{itemize}
\item We sample $\bS = [\bS_1, \dots, \bS_n] \sim \calS$, and we give the algorithm as input $\bS$, $\sum_{i=1}^n \bS_i (\lambda_i^{1/p} x_i)$, and $\lambda_1,\dots, \lambda_n$. 
\item The algorithm outputs a tuple of three numbers $(\bj, \balpha, \bbeta) \in [d] \times \R_{\geq 0} \times \R_{\geq0}$. With probability at least $1 - \delta$ over the draw of $\bS \sim \calS$, we have the following two inequalities:
\begin{align*}
(1-\eps) \left(\sum_{i=1}^n \lambda_i |x_{i \bj}|^p\right)^{1/p} &\leq \balpha \leq (1+\eps) \left( \sum_{i=1}^n \lambda_i |x_{i \bj}|^p \right)^{1/p}, \\
(1-\eps) \min_{z \in \R} \left( \sum_{i=1}^n \lambda_i |x_{i\bj} - z|^p\right)^{1/p} &\leq \bbeta \leq (1+\eps) \min_{z \in \R} \left( \sum_{i=1}^n \lambda_i |x_{i\bj} - z|^p \right)^{1/p}.
\end{align*}
\item Furthermore, the distribution of the random variable $\bj$ is $\eps 2^{-p}$-close in total variation distance to $\calD$. 
\end{itemize}
\end{lemma}

\newcommand{\minmax}{\mathrm{minmax}}

\begin{proof}[Proof of Theorem~\ref{thm:median-sketch} assuming Lemma~\ref{lem:main-sketch-lemma}]
Given Lemma~\ref{lem:main-sketch-lemma}, the proof of Theorem~\ref{thm:median-sketch} is straight-forward. We fix $\lambda_1 = \dots = \lambda_n = 1/n$, and we first handle the centering. We will utilize Lemma~\ref{lem:main-sketch-lemma} which requires vectors $x_1,\dots, x_n$ to satisfy $\sum_{i=1}^n \lambda_i x_i = 0$; hence, we sketch the vectors $x_1', \dots, x_n'$ given by $x_i' = x_i - \sum_{h=1}^n \lambda_h x_h$, which are now centered. By linearity, this is equivalent to maintaining the vector
\[ \sum_{i=1}^n \lambda_i^{1/p} \bS_i(x_i - \sum_{h=1}^n \lambda_h x_h) = \sum_{i=1}^n \left( \lambda_i^{1/p} \bS_i - \lambda_i \sum_{h=1}^n \lambda_h^{1/p} \bS_h \right) x_i \in \R^s. \]
We take $t = \omega(1/\eps^2)$ independent sketches from Lemma~\ref{lem:main-sketch-lemma} with accuracy parameter $\eps/2$ and error probability $\delta = o(1/t)$. This, in turn, gives us $t$ independent samples $(\bj_1, \balpha_1, \bbeta_1), \dots, (\bj_{t}, \balpha_1, \bbeta_t)$. By taking a union bound over the $t$ executions of Lemma~\ref{lem:main-sketch-lemma}, with high probability, every $\balpha_1,\dots, \balpha_t$ and $\bbeta_1,\dots, \bbeta_t$ satisfy 
\[ \balpha_{\ell}^p \approx_{(1+\eps p/2)} \sum_{i=1}^n \lambda_i |x_{i\bk}|^p \qquad\text{and}\qquad \bbeta_{\ell}^p \approx_{(1+\eps p/2)} \min_{z \in \R} \sum_{i=1}^n \lambda_i |x_{i\bk} - z|^p,   \]
and $\bj_{1}, \dots, \bj_{t}$ are independent draws from a distribution $\calD'$ which is $\eps 2^{-p}$-close to $\calD$. For estimating $Z$, we use an $\ell_p$-sketch to accuracy $\eps/2$ and failure probability $\delta = o(1)$. For example, the sketch for $Z$ may proceed by applying an $\ell_p$-sketch \cite{I06} to the stacked vector $x' \in \R^{nd}$ where
\[x_{ij}' = \lambda_i^{1/p} \cdot x_{ij},\]
so that the $\ell_p$ norm of $x'$ is exactly $Z^{1/p}$. Let $\hat{\bZ}$ be the estimate for the $Z$. For our estimate $\boldeta$ that we will output, we set
\[ \boldeta = \hat{\bZ} \cdot \frac{1}{t} \sum_{\ell = 1}^t \minmax\left\{2^{-p}, \left(\frac{\bbeta_{\ell}}{\balpha_{\ell}}\right)^p, 1 \right\},\]
where $\minmax(l, x, u)$ is $l$ if $x \leq l$, $u$ if $u \geq x$, and $x$ otherwise. To see why our estimator approximates (\ref{eq:quantity-to-estimate}), we have $\hat{\bZ}$ is a $(1\pm \eps/2)$-approximation of $Z$. The latter quantity is the empirical average of $t$ i.i.d random variables, each of which is bounded by $2^{-p}$ and $1$. In particular, we have that with probability at least $1 -o(1)$, Chebyshev's inequality, and the conditions of $\bbeta_{\ell}$ and $\balpha_{\ell}$, 
\[ \Ex_{\bj \sim \calD'}\left[\dfrac{\min_{z \in \R} \sum_{i=1}^n \lambda_i |x_{i\bj} - z|^p}{\sum_{i=1}^n \lambda_i |x_{i\bj}|^p}  \right] \approx_{(1+2\eps p)} \frac{1}{t} \sum_{\ell=1}^t \minmax\left\{ 2^{-p}, \left( \frac{\bbeta_{\ell}}{\balpha_{\ell}}\right)^p, 1 \right\}. \]
It remains to show that
\[ \Ex_{\bj \sim \calD'}\left[\dfrac{\min_{z \in \R} \sum_{i=1}^n \lambda_i |x_{i\bj} - z|^p}{\sum_{i=1}^n \lambda_i |x_{i\bj}|^p}  \right]  \approx_{(1\pm \eps)} \Ex_{\bj \sim \calD}\left[\dfrac{\min_{z \in \R} \sum_{i=1}^n \lambda_i|x_{i\bj} - z|^p}{\sum_{i=1}^n \lambda_i |x_{i\bj}|^p}  \right]. \] 
This follows from two facts: (1) $\calD'$ and $\calD$ are $\eps 2^{-p}$ close, since the random variable is at most $1$, the expectations are off by at most an additive $\eps 2^{-p}$-factor, and (2) both quantities above are the average of random variables which are at least $2^{-p}$, so an additive $\eps 2^{-p}$ error is less than a multiplicative $(1\pm \eps)$-error.

The above gives an estimate which is a $1\pm \eps$-approximation with probability $1 - o(1)$, in order to boost the probability of success to $1 - \delta$, we simply repeat $O(\log(1/\delta))$ times and output the median estimate.
\end{proof}

The remainder of the section is organized as follows. We give in the (next) Subsection~\ref{subsec:single-coord}, the necessary sketches for obtaining $\balpha$ and $\bbeta$ for a fixed coordinate $j$. Then, in the following Subsection~\ref{subsec:grouping}, we show how we combine various sketches from Subsection~\ref{subsec:single-coord} for different $j \in [d]$ to obtain $\balpha$ and $\bbeta$ up to some additive error. Finally, the proof of Lemma~\ref{lem:main-sketch-lemma} appears in Subsection~\ref{subsec:proof-of-lemma}, where we apply a randomized transformation to the input so that the additive error from Subsection~\ref{subsec:grouping} is a multiplicative error for the specific sampled $\bj$. 

%% file: single-coordinate.tex

\subsection{Sketch for Optimizing a Single Coordinate}\label{subsec:single-coord}

\newcommand{\err}{\mathrm{err}}

In this subsection, we give linear sketches which are useful for optimizing over a single coordinate. Specifically, given the $n$ vectors $x_1,\dots, x_n \in \R^d$ and $j \in [d]$, we consider the $k$-th coordinate of the $n$ vectors $x_{1j}, x_{2j},\dots, x_{n j} \in \R$. Hence, the linear sketches in this section will act on vectors in $\R^n$, corresponding to the $j$-th coordinates of the points, and will give approximations to 
\[ \sum_{i=1}^n \lambda_j |x_{ij}|^p \quad \text{(Corollary~\ref{cor:cross-ps})}\quad\text{and} \qquad \min_{y_j \in \R} \sum_{i=1}^n \lambda_j |y_j - x_{ij}|^p. \quad \text{(Lemma~\ref{lem:single-coord-opt})}\]
The lemma statements also consider an additive error term, $\err \in \R_{\geq0}$, which will be necessary when combining these sketches in Subsection~\ref{subsec:grouping}; however, it may be helpful to consider $\err = 0$ on first reading. 

\begin{lemma}\label{lem:one-y}
For any $n \in \N$, $p \in [1,2]$ and $\eps, \delta \in (0, 1)$, let $s = O(\log(1/\delta) / \eps^2)$. There exists a distribution $\calM$ over $s \times n$ matrices, and an algorithm such that for any $x \in \R^n$, $\lambda_1,\dots, \lambda_n \in [0,1]$, and $y \in \R$, the following occurs:
\begin{itemize}
\item We sample $\bS \sim \calM$ as well as a random vector $\bchi = (\bchi_1,\dots, \bchi_s) \in \R^s$ where each is an i.i.d $p$-stable random variable. For any $\err \in \R_{\geq 0}$, we give the algorithm as input $\bS$, $\bS (\lambda^{1/p} \circ x) + \err \cdot\bchi $, the parameters $\lambda_1,\dots, \lambda_n$, and $y$.\footnote{The notation $\lambda^{1/p} \circ x \in \R^n$ denotes the Hadamard product, where $(\lambda^{1/p} \circ x)_i = \lambda_i^{1/p} \cdot x_i$. }
\item The algorithm outputs a parameter $\hat{\boldeta} \in \R_{\geq 0}$, which depends on $\bS, \bS (\lambda^{1/p} \circ x) + \err \cdot \bchi$, the parameters $\lambda_1, \dots, \lambda_n$, and $y$ which satisfies with probability at least $1 - \delta$ over $\bS$ and $\bchi$,
\begin{align}
(1-\eps) \left(\sum_{i=1}^n \lambda_i |x_i - y|^p + \err^p \right)^{1/p} \leq \hat{\boldeta} \leq (1+\eps) \left(\sum_{i=1}^n \lambda_i |x_i - y|^p + \err^p\right)^{1/p} . \label{eq:eta-guarantee}
\end{align} 
\end{itemize}
Furthermore, for every $j \in [s]$, the random variables $(\bS (\lambda^{1/p} \circ x))_j \in \R$ are distributed as $\| \lambda^{1/p} \circ x\|_p \cdot \bchi_j$, where $\bchi_j$ are independent, $p$-stable random variables.
\end{lemma}

\begin{proof}
We notice that this simply corresponds to an $\ell_p$-sketch of the vector $z \in \R^n$, which is given by letting each $z_i = \lambda^{1/p} \circ (x_i - y)$, so that the $\ell_p$-sketch of \cite{I06} would accomplish this task. Since the algorithm receives $\bS$, $\bS(\lambda^{1/p} \circ x) + \err \cdot \chi$ and $y \in \R$, the algorithm may compute $\bS (\lambda^{1/p} \circ y \cdot \ind)$, where $\ind \in \R^n$ is an all-1's vector, and evaluate the sketch 
\[ \bS(\lambda^{1/p} \circ x) + \err \cdot \bchi - \bS (\lambda^{1/p} \circ y \cdot \ind) = \bS (\lambda^{1/p} \circ (x - y\cdot \ind)) + \err \cdot \bchi\] 
by linearity. Furthermore, note that the error simply corresponds to an $\ell_p$-sketch of the vector $z' \in \R^{n+1}$ which sets $z'_i = z_i$ for $i \neq n+1$ and $z_{n+1}' = \err$.
\end{proof}

\begin{corollary}\label{cor:cross-ps}
For any $n \in \N$, $p \in [1, 2]$ and $\eps, \delta \in (0,1)$, let $s = O(\log(1/\delta)/\eps^2)$. There exists a distribution $\calM$ over $s \times n$ matrices, and an algorithm such that for any $x \in \R^n$, and any $\lambda_1,\dots, \lambda_n \in [0,1]$, the following occurs:
\begin{itemize}
\item We sample $\bS \sim \calM$ and a random vector $\bchi = (\bchi_1,\dots, \bchi_{s}) \in \R^s$ of i.i.d $p$-stable random variables. For any $\err \in \R_{\geq 0}$. We give the algorithm as input $\bS$ and $\bS (\lambda^{1/p} \circ x) + \err \cdot \bchi$.
\item With probability at least $1 -\delta$ over $\bS$ and $\bchi$, the algorithm outputs a parameter $\hat{\bgamma} \in \R_{\geq 0}$, which depends on $\bS$, and $\bS (\lambda^{1/p} \circ x)$ which satisfies
\begin{align*}
(1-\eps) \left(\sum_{i=1}^n \lambda_i|x_i|^p + \err^p\right)^{1/p} \leq \hat{\bgamma} \leq (1+\eps) \left(\sum_{i=1}^n \lambda_i |x_i|^p + \err^p\right)^{1/p}. 
\end{align*} 
\end{itemize}
Furthermore, for every $j \in [s]$, the random variables $(\bS (\lambda^{1/p} \circ x))_j\in \R$ are independent and distributed as $\| \lambda^{1/p} \circ x\|_p\cdot \bchi_j$, where $\bchi_j$ are independent, $p$-stable random variables.
\end{corollary}

\begin{proof}
We apply Lemma~\ref{lem:one-y} to the vector $x \in \R^{n}$ with $y = 0$.
\end{proof}

\begin{lemma}\label{lem:single-coord-opt}
For any $n \in \N$, $p \in [1,2]$ and $\eps, \delta \in (0, 1)$, let $s = O(\log(1/(\eps \delta))/\eps^2)$. There exists a distribution $\calM$ over $s \times n$ matrices, and an algorithm such that for any $x \in \R^n$ and any $\lambda_1,\dots, \lambda_n \in [0,1]$ with $\sum_{i=1}^n \lambda_i = 1$, whenever $\sum_{i=1}^n \lambda_i x_i = 0$, the following occurs:
\begin{itemize}
\item We sample $\bS \sim\calM$ and a random vector $\bchi =(\bchi_1,\dots, \bchi_s)$ of i.i.d $p$-stable random variables. For any $\err \in \R_{\geq 0}$. We give the algorithm as input $\bS$, $\bS (\lambda^{1/p}\circ x) + \err\cdot \bchi$, the parameters $\lambda_1,\dots, \lambda_n$, and a parameter $\gamma \in \R_{\geq 0}$ satisfying
\begin{align*}
(1-\eps) \left(\sum_{i=1}^n \lambda_i |x_i|^p\right)^{1/p} \leq \gamma \leq (1+\eps) \left(\sum_{i=1}^n \lambda_i |x_i|^p\right)^{1/p}.
\end{align*}
\item The algorithm outputs a parameter $\hat{\bbeta} \in \R_{\geq0}$ which satisfies 
\begin{align} 
(1-\eps) \min_{z \in \R} \left(\sum_{i=1}^n \lambda_i |x_i - z|^p + \err^p\right)^{1/p} \leq \hat{\bbeta} \leq (1+\eps) \min_{z\in \R} \left(\sum_{i=1}^n \lambda_i |x_i - z|^p + \err^p\right)^{1/p}.  \label{eq:minimize}
\end{align}
\end{itemize}
Furthermore, for every $j \in [s]$, the random variables $(\bS (\lambda^{1/p} \circ x))_j$ are independent and distributed as $\| \lambda^{1/p} \circ x\|_p \cdot \bchi_j$, where $\bchi_j$ are independent, $p$-stable random variables.
\end{lemma}

\begin{proof}
We will utilize the sketch from Lemma~\ref{lem:one-y}, while varying the $y$'s to find the minimum. Specifically, let $t = 16 \cdot 2^p / \eps$, and let the distribution $\calM$ be the same as that of Lemma~\ref{lem:one-y} instantiated with error probability $1 - t\delta$ and accuracy parameter $\eps / 2$. We discretize the interval $[-4\gamma,4\gamma]$ into $t$, evenly-spaced out points $y_1, \dots, y_t \subset [-4\gamma, 4\gamma]$ such that $y_{\ell+1} -y_{\ell} = 8\gamma / t$. We utilize the algorithm in Lemma~\ref{lem:one-y} to obtain estimates $\hat{\boldeta}_{1},\dots, \hat{\boldeta}_t$ satisfying (\ref{eq:eta-guarantee}) with $y_1, \dots,y_t$, respectively. Then, we output
\[ \hat{\bbeta} = \min_{\ell \in [t]} \hat{\boldeta}_{\ell}. \]
Since we amplified the error probability to less than $t\delta$, we may assume, by a union bound, that all estimates $\{\boldeta_{\ell}\}_{\ell \in [t]}$ satisfy (\ref{eq:eta-guarantee}) with $y_{\ell}$ with probability at least $1-\delta$. First, for any $\ell \in [t]$, 
\begin{align*}
\min_{z \in \R} \left( \sum_{i=1}^n \lambda_i |x_i - z|^p \right)^{1/p} \leq \left( \sum_{i=1}^n \lambda_i |x_i - y_{\ell}|^p \right)^{1/p}, 
\end{align*}
and therefore, the lower bound in (\ref{eq:minimize}) is implied by (\ref{eq:eta-guarantee}). To prove the upper bound in (\ref{eq:minimize}), denote $z \in \R$ as the true minimizer of $(\sum_{i=1}^n \lambda_i |x_i - z|^p)^{1/p}$. By the fact $\sum_{i=1}^n \lambda_i = 1$ and the triangle inequality, we have
\begin{align*}
|z| = \left( \sum_{i=1}^n \lambda_i |z|^p\right)^{1/p} \leq \left( \sum_{i=1}^n \lambda_i |x_i - z|^{p} \right)^{1/p} + \left( \sum_{i=1}^n \lambda_i |x_i|^p\right)^{1/p} \leq 2 \left( \sum_{i=1}^n \lambda_i |x_i|^p\right)^{1/p},
\end{align*}
so $|z| \leq 2 (1+\eps) \gamma \leq 4\gamma$, and thus $z \in [-4\gamma, 4\gamma]$. Let $\ell \in [t]$ be such that $|y_t - z| \leq 4\gamma / t$. Then, again by the triangle inequality and the fact $\sum_{i=1}^n \lambda_i x_i = 0$,
\begin{align*}
\left( \sum_{i=1}^n \lambda_i |x_i - y_t|^p\right)^{1/p} \leq \left( \sum_{i=1}^n \lambda_i |x_i - z|^p\right)^{1/p} + 4 \gamma / t \leq \left(1 + 4(1+\eps) 2^p / t\right)\left(\sum_{i=1}^n \lambda_i |x_i - z|^p \right)^{1/p}.
\end{align*}
By the setting of $t$, $(\sum_{i=1}^n \lambda_i |x_i - y_t|^p)^{1/p} \leq (1+\eps/2)(\sum_{i=1}^n \lambda_i |x_i - z|^p)^{1/p}$, and by (\ref{eq:eta-guarantee}), we obtain the desired upper bound.
\end{proof}

%% file: sampling-with-meta-data.tex

\subsection{Grouping Single Coordinate Sketches}\label{subsec:grouping}

\newcommand{\Exp}{\mathrm{Exp}}
\newcommand{\bolderr}{\boldsymbol{\err}}

In this subsection, we show how to compress $d$ linear sketches (one for each coordinate) from Subsection~\ref{subsec:single-coord}. In the lemma that follows, the parameter $m \in \N$ should be considered the sketch size of the sketches in Subsection~\ref{subsec:single-coord}, and the linear sketch will take the $d$ sketches from Subsection~\ref{subsec:single-coord} (represented as a vector $\R^{dm}$). Each of the $d$ linear sketches have each coordinate of $\R^m$ distributed as an i.i.d scaled $p$-stable random variable (specified by the last sentence in Corollary~\ref{cor:cross-ps} and Lemma~\ref{lem:single-coord-opt}). Thus, we write the $d$ sketches as $\bPsi_1 v_1,\dots, \bPsi_d v_d \in \R^m$, where $v_j \in \R$ is a scaling, and $\bPsi_1,\dots, \bPsi_d \in \R^{m \times n}$ are i.i.d $p$-stable matrices. 

\begin{lemma}[$p$-stable Sketch Compression via Count-Min]\label{lem:count-sketch}
Let $d,m \in \N$, $\eps, \delta \in (0, 1)$, and let $t = O(\log(d/\delta))$. There exists a distribution $\calC$ over $(10tm/\eps^{p}) \times (dm)$ matrices, and an algorithm such that for any $v \in \R^{d}$, the following occurs:
\begin{itemize}
\item We sample $\bC \sim \calC$ and a $(dm) \times d$ matrix $\bPsi$, where $\bPsi = \diag(\bPsi_1,\dots, \bPsi_d)$, and each $\bPsi_j = (\bchi_{j1}, \dots, \bchi_{jm}) \in \R^m$ are independent $p$-stable random vectors.\footnote{Hence, the vector $\bPsi v \in \R^{dm}$ is given by vertically stacking $d$ vectors of the form $v_j \bPsi_j \in \R^m$.} The algorithm receives as input $\bC$ and $\bC \bPsi v$. 
\item The algorithm outputs, for each $j \in [d]$, a sequence of $t$ vectors $\hat{\bz}_j^{(1)}, \dots, \hat{\bz}_{j}^{(t)} \in \R^m$ which satisfy, for each $t' \in [t]$,
\begin{align}
\hat{\bz}_{j}^{(t')} = v_j \bPsi_j + \bolderr_j^{(t')} \cdot \bchi_j^{(t')}.\label{eq:lem-z-output}
\end{align}
where $\bchi_{j}^{(t')} \in \R^{m}$ is a vector of independent $p$-stable random variables, and $\bolderr_j^{(t')} \in \R_{\geq 0}$ only depends on $\bC$. With probability at least $1 - \delta$ over $\bC$, for every $j \in [d]$
\begin{align} 
\left| \left\{ t' \in [t] : \bolderr_{j}^{(t')} \leq \eps \|v\|_p \right\} \right| \geq t/2. \label{eq:error-bound}
\end{align}
\end{itemize}
\end{lemma}

\begin{proof}
The matrix $\bC$ is a Count-Min matrix which given a vector $u \in \R^{dm}$ given by vertically stacking $d$ vectors $u_1,\dots, u_d \in \R^m$ repeats the following process: for each $t' \in [t]$, we sample a hash function $\bh_{t'} \colon [d] \to [10/\eps^p]$, and for each $\ell \in [10/\eps^p]$ store the vector
\[ \bb_{t', \ell} \eqdef \sum_{j \in [d]} \ind\{ \bh_{t'}(j) = \ell \} \cdot u_{j} \in \R^m. \]
In particular, the output $\bC u$ consists of stacking $t \cdot 10 / \eps^p$ vectors $(\bb_{t', \ell} \in \R^m : t' \in [t], \ell \in [10/\eps^p])$, which gives the desired bound of $10 m t / \eps^{p}$ on the output dimension of $\bC$. For each $j \in [d]$ and $t' \in [t]$, the algorithm lets $\bell = \bh_{t'}(j)$ and sets
\[ \hat{\bz}_j^{(t')} = \bb_{t', \ell} = v_j \bPsi_j + \sum_{j' \in [d] \setminus \{j\}} \ind\{ \bh_{t'}(j') = \bell \} \cdot v_{j'} \cdot \bPsi_{j'}. \]
We now apply the $p$-stability property to the right-most summand, to notice that 
\[ \bolderr_{j}^{(t')} = \left(\sum_{j' \in [d] \setminus \{j\}} \ind\{ \bh_{t'}(j') = \bell\} \cdot v_{j'}^p \right)^{1/p},\]
which only depends on $\bC$. Furthermore, the inner most summand is at most $\eps^p/10 \sum_{j' \in [d] \setminus \{j\}} v_{j'}^p$ in expectation. By Markov's inequality, each $\bolderr_{j}^{(t')} \leq \eps \| v\|_p$ with probability at least $9/10$. Since $t = O(\log(d / \delta))$, the probability that (\ref{eq:error-bound}) is not satisfied for each $j \in [d]$ is at most $\delta / d$ by a Chernoff bound, so that a union bound gives the desired guarantees.
\end{proof}

The above lemma allows us to compress $d$ many $p$-stable sketches into $O(\log(d/\delta)/\eps^{p})$ many $p$-stable sketches, albeit with some error. Since the $p$-stable sketches that we will use (from Corollary~\ref{cor:cross-ps} and Lemma~\ref{lem:single-coord-opt}) are exactly of the form $\bPsi v$ for some vector $v$, Lemma~\ref{lem:count-sketch} will allow us to compress them. Namely, we will consider $d$ sketches from Corollary~\ref{cor:cross-ps} and Lemma~\ref{lem:single-coord-opt} and utilize Lemma~\ref{lem:count-sketch}; for each $j \in [d]$, we will be able to recover $t$ noisy versions of the sketch of Corollary~\ref{cor:cross-ps} and Lemma~\ref{lem:single-coord-opt} for coordinate $j$. Importantly, the noise is of the form an error times a $p$-stable random variable, and these are the kinds of errors that Corollary~\ref{cor:cross-ps} and Lemma~\ref{lem:single-coord-opt} can easily handle. 

\begin{lemma}[$p$-stable Sketch Recovery for Sample]\label{lem:p-stable-recovery}
For $n, d \in \N$, $p \in [1,2]$ and $\eps, \delta \in (0, 1)$, let $s = O(\log^2(d/\delta) /\eps^{2+p})$. There exists a distribution $\calR$ over $s \times (nd)$ matrices, and an algorithm such that for any vectors $y_1, \dots, y_n \in \R^d$ and weights $\lambda_1,\dots, \lambda_n \in [0,1]$ with $\sum_{i=1}^n \lambda_i = 1$, the following occurs with probability at least $1 - \delta$:
\begin{itemize}
\item We sample $\bS = [\bS_1, \dots, \bS_n ] \sim \calR$ and we give the algorithm as input $\bS$, and the vector $\sum_{i=1}^n \bS_i (\lambda_i^{1/p} y_i) \in \R^{s}$.
\item The algorithm outputs $d$ numbers $\balpha_1, \dots, \balpha_d \in \R_{\geq 0}$ such that each $j \in [d]$ satisfies
\[ (1-\eps)\left( \sum_{i=1}^n \lambda_i |y_{ij}|^p\right)^{1/p} - \err \leq \balpha_j \leq (1+\eps)\left( \sum_{i =1}^n \lambda_i |y_{ij}|^p \right)^{1/p} + \err,\]
where $\err \in \R_{\geq0}$ is an additive error satisfying
\[ \err \leq \eps \left(\sum_{j=1}^d \sum_{i=1}^n \lambda_i |y_{ij}|^p \right)^{1/p}. \] 
\end{itemize}
\end{lemma}

\begin{proof}
We combine Corollary~\ref{cor:cross-ps} and Lemma~\ref{lem:count-sketch}. We describe the distribution $\calR$ over $s \times (nd)$ matrices by giving a procedure for sampling $\bS \sim \calR$. $\bS$ can be naturally expressed as a concatenation of matrices $\bS = [\bS_1, \dots, \bS_n]$.
\begin{itemize}
\item We let $\calM$ be the distribution over $s_0 \times n$ matrices of Corollary~\ref{cor:cross-ps} with error probability at most $\delta/(2dt)$ and accuracy $\eps$ (so that we may union bound over $d$ sketches later) so that $s_0 = O(\log(dt/\delta) / \eps^2)$. We take $d$ independent samples $\bS_1',\dots,\bS_d' \sim \calM$.
\item For each $j \in [d]$, we let $P_j$ be the $n \times (nd)$ matrix where given the vector $y' \in \R^{nd}$ given by vertically stacking $\lambda_1^{1/p} y_1,\dots, \lambda_n^{1/p} y_n \in \R^d$, sets $y_{\cdot, j}' = P_k j'$, where $y_{\cdot, j}' = \lambda^{1/p} \circ (y_{i, j})_{i\in[n]} \in \R^n$. Let $P$ be the $(nd) \times (nd)$ matrix which stacks these matrices vertically. 
\item We sample $\bC \sim \calC$ as in Lemma~\ref{lem:count-sketch} with $m = s_0$, where we set the accuracy parameter $\eps/2$ and the failure probability $\delta/2$. We let 
\[ \bS = \bC \cdot \diag(\bS_1',\dots, \bS_d') \cdot P. \]
\end{itemize}
Intuitively, we will apply our sketch $\bS$ on the vector the matrix $\bS$ may be interpreted as first applying $d$ sketches of Corollary~\ref{cor:cross-ps} to the vectors $(\lambda^{1/p} \circ y_{\cdot, 1}), \dots, (\lambda^{1/p} \circ y_{\cdot, d}) \in \R^n$, and then applying $\bC$ from Lemma~\ref{lem:count-sketch}. The algorithm for producing the estimates $\balpha_1,\dots,\balpha_d$ proceeds by applying the algorithm of Lemma~\ref{lem:count-sketch} to obtain, for each $j \in [d]$ a sequence of $t$ vectors $\hat{\bz}_{j}^{(1)},\dots, \hat{\bz}_{j}^{(t)} \in \R^{s_0}$. We apply the algorithm of Corollary~\ref{cor:cross-ps} to each of the $t$ vectors to obtain estimates $\balpha_j^{(1)},\dots, \balpha_j^{(t)} \in \R_{\geq 0}$, and we let $\balpha_j = \median \{ \balpha_{j}^{(t')}: t' \in [t] \}$.

To see why this works, consider the collection of $d$ vectors
\[ \bz_j = \bS_j' (\lambda^{1/p} \circ y_{\cdot, j}) \in \R^{s_0},\]
and notice that by Corollary~\ref{cor:cross-ps}, every $j\in [d]$ and $\ell \in [s]$, $\bz_{j,\ell} \sim (\sum_{i=1}^n \lambda_i |y_{i,j}|^p )^{1/p} \cdot \bchi_{j,\ell}$, where $\bchi_{k,\ell}$ are independent, $p$-stable random variables. Indeed, if we write $v \in \R^{d}$ as the vector which sets
\[ v_{j} = \left(\sum_{i=1}^n \lambda_i |y_{i,j}|^p\right)^{1/p},\]
then vertically stacking the vectors $\bz_1,\dots, \bz_d \in \R^{s_0}$ gives a vector which is equivalently distributed as $\bPsi v$, where $\bPsi$ is the matrix from Lemma~\ref{lem:count-sketch}. In particular, with probability at least $1 - \delta / 2$, the algorithm of Lemma~\ref{lem:count-sketch} outputs $dt$ vectors $(\hat{\bz}_j^{(t')} : j \in [d], t' \in [t])$ which satisfy 
\begin{align}
\hat{\bz}_j^{(t')} = \bS_k' (\lambda^{1/p}\circ y_{\cdot, j}) + \bolderr_{j}^{(t')} \cdot \bchi_{j, t'}.
\end{align}
Hence, with probability at least $1 - \delta / (2dt)$, the algorithm of Corollary~\ref{cor:cross-ps} applied to $\hat{\bz}_{j}^{(t')}$ outputs an estimate $\balpha_{j}^{(t')}$ satisfying $(1-\eps) ( v_j^p + (\bolderr_j^{(t')})^p )^{1/p} \leq \balpha_j^{(t')} \leq (1+\eps) ( v_j^p + (\bolderr_j^{(t')})^p )^{1/p}$, and therefore, we have that each $\balpha_j^{(t')}$ satisfies
\[ (1-\eps) v_{j} - \bolderr_j^{(t')} \leq \balpha_{j}^{(t')} \leq (1+\eps) v_j + 2 \cdot \bolderr_j^{(t')} .\]
Since at least $t/2$ of $t' \in [t]$ satisfies $\bolderr_{j}^{(t')} \leq \eps/2 \cdot \| v\|_p$, the median $\balpha_j^{(t')}$ satisfies the desired error guarantee. Applying a union bound over all $dt$ applications of Corollary~\ref{cor:cross-ps} and Lemma~\ref{lem:count-sketch} gives the desired guarantees.
\end{proof}

\begin{lemma}[$p$-stable Sketch Recovery for Optimizer]\label{lem:p-stable-recovery-opt}
For $n, d \in \N$, $p \in [1,2]$ and $\eps, \delta \in (0,1)$, let $s = O(\log(d/\delta) \cdot \log(\log d/(\eps \delta)) /\eps^{2+p})$. There exists a distribution $\calO$ over $s \times (nd)$ matrices, and an algorithm such that for any vectors $y_1,\dots, y_n \in \R^d$ and any set of weights $\lambda_1,\dots, \lambda_n \in [0,1]$ where $\sum_{i=1}^n \lambda_i = 0$, whenever $\sum_{i=1}^n \lambda_i y_i = 0$, the following occurs with probability at least $1 - \delta$:
\begin{itemize}
\item We sample $\bS = [\bS_1,\dots, \bS_n ]\sim \calO$ and we give the algorithm as input $\bS$, $\sum_{i=1}^n \bS_i (\lambda_i^{1/p} y_i)$, the parameters $\lambda_1,\dots, \lambda_n$, an index $j_0 \in [d]$, and a parameter $\gamma \in \R_{\geq 0}$ satisfying
\[ (1-\eps) \left( \sum_{i=1}^n \lambda_i |y_{ij_0}|^p\right)^{1/p} \leq \gamma \leq (1+\eps) \left( \sum_{i=1}^n \lambda_i |y_{ij_0}|^p\right)^{1/p}.\]
\item The algorithm outputs a parameter $\hat{\bbeta} \in \R_{\geq 0}$ which satisfies
\[ (1-\eps) \min_{z \in \R} \left( \sum_{i=1}^n \lambda_i |y_{ij_0} - z|^p\right)^{1/p} - \err \leq \hat{\bbeta} \leq (1+\eps) \min_{z \in \R} \left( \sum_{i=1}^n \lambda_i |y_{ij_0} - z|^p\right)^{1/p} + \err, \]
where $\err \in \R_{\geq 0}$ is an additive error satisfying
\[ \err \leq \eps \left(\sum_{j=1}^d \sum_{i=1}^n \lambda_i|y_{i, j}|^p \right)^{1/p}. \]
\end{itemize}
\end{lemma}

\begin{proof}
The proof follows similarly to that of Lemma~\ref{lem:p-stable-recovery}; the only difference is that instead of using the sketch of Corollary~\ref{cor:cross-ps}, we use the sketch of Lemma~\ref{lem:single-coord-opt}. For completeness, we describe the distribution $\calO$ over $s \times (nd)$ matrices by giving a procedure for sampling $\bS \sim \calO$:
\begin{itemize}
\item We let $\calM$ be the distribution over $s_0 \times n$ matrices from Lemma~\ref{lem:single-coord-opt} with accuracy $\eps$ and failure probability $\delta /(2t)$, where $s_0 = O(\log(t/(\eps \delta))/\eps^2)$. We take $d$ independent samples $\bS_1', \dots, \bS_d' \sim \calM$. Note that even though we take $d$ independent samples, we will only require that the sketch $t$ evaluations of the $\bS_{j_0} \sim \calM$ succeed (hence, we amplify the error probability to $\delta / (2t)$, as opposed to $\delta / (2td)$ as in Lemma~\ref{lem:p-stable-recovery}).
\item We sample $\bC \sim \calC$ as in Lemma~\ref{lem:count-sketch} with $m = s_0$, where we set the accuracy parameter $\eps/2$ and failure probability $\delta /2$. Recalling the definition of $P$ (see Item 2 in the proof of Lemma~\ref{lem:p-stable-recovery}, we let
\[ \bS = \bC \cdot \diag(\bS_1', \dots, \bS_d') \cdot P. \]
\end{itemize}
Similarly to the proof of Lemma~\ref{lem:p-stable-recovery}, $\bS$ may be interpreted as applying the sketch of Lemma~\ref{lem:count-sketch} to $d$ vectors in $\R^{s_0}$, each $j \in [d]$ of which is an independent sketch $\bS_j' (\lambda^{1/p}\circ y_{\cdot j}) \in \R^{s_0}$, where $\bS_j' \sim \calM$ is the sketch of Lemma~\ref{lem:single-coord-opt}. Again, we consider the collection of $d$ vectors $\bz_j = \bS_j' (\lambda^{1/p} \circ y_{\cdot, j}) \in \R^{s_0}$, for all $j \in [d]$, and by Lemma~\ref{lem:single-coord-opt}, every $j \in [d]$ has $\bz_{j} \sim \|\lambda^{1/p} \circ y_{\cdot, j}\|_p \cdot \bPsi_{j} \in \R^{s_0}$, where $\bPsi_{j}$ is an independent, $p$-stable random vector. Writing $v \in \R^{d}$ by $v_{j} = \|\lambda^{1/p} \circ y_{\cdot, j}\|_p$, and we apply the algorithm of Lemma~\ref{lem:count-sketch} and focus on the $t$ vectors $\hat{\bz}_{j_0}^{(1)},\dots, \hat{\bz}_{j_0}^{(t)} \in \R^{s_0}$ which satisfy
\[ \hat{\bz}_{j_0}^{(t')} = \bS_{j_0}' (\lambda^{1/p} \circ y_{\cdot, j_0}) + \bolderr_{j_0}^{(t')} \cdot \bchi_{j, t'}. \]
We apply the algorithm of Lemma~\ref{lem:single-coord-opt} to each of the vectors $\hat{\bz}_{j_0} \in \R^{s}$, while giving as input the parameter $\gamma$ to obtain the estimate $\hat{\bbeta}^{(1)}, \dots, \hat{\bbeta}^{(t)}$. Then, we set $\hat{\bbeta} = \median \{ \hat{\bbeta}^{(t')} : t' \in [t] \}$. Similarly to the proof of Lemma~\ref{lem:p-stable-recovery}, $\hat{\bbeta}$ provides the desired approximation guarantees.
\end{proof}

\subsection{Proof of Lemma~\ref{lem:main-sketch-lemma}}\label{subsec:proof-of-lemma}

The distribution over matrices $\calS$ will be given by utilizing the sketches in Lemmas~\ref{lem:p-stable-recovery} and~\ref{lem:p-stable-recovery-opt}. In particular, let $\calR$ denote the distribution over $s_1 \times (nd)$ matrices from Lemma~\ref{lem:p-stable-recovery} where $s_1 = O(\log^2(d/(\delta_1 \eps_1)) / \eps_1^{2+p})$ using failure probability at most $\delta_1 = \delta \eps 2^p/ 3$ and accuracy $\eps_1$, and let $\calO$ denote the distribution over $s_2 \times (nd)$ matrices from Lemma~\ref{lem:p-stable-recovery-opt} where $s_2 = O(\log(d/\delta_2)\log(\log d/(\eps_2 \delta_2))/\eps_2^{2+p})$ using failure probability at most $\delta_2 = \delta /3$ and accuracy $\eps_2$. Specifically, we let
\[ \eps_1 = c_1 \cdot \eps^{1 + 2/p} \cdot \delta^{1+1/p} \qquad\text{and} \qquad \eps_2 = c_2 \cdot \eps^{1+1/p} \cdot \delta^{1/p},\]
for small enough constant $c_1, c_2 > 0$. Thus, the total space
\begin{align}
s = s_1 + s_2 = O\left(\frac{\log^2(d/(\eps\delta))}{\eps^{4 + p + 4/p} \delta^{3 + p + 2/p}} \right) + O\left( \frac{\log(d/\delta) \log(\log d/(\eps \delta))}{\eps^{3 + p + 2/p}} \right). \label{eq:space-complexity}
\end{align}
A draw $\bS = [\bS_1, \dots, \bS_n] \sim \calS$ is an $(s_1 + s_2) \times (nd)$ matrix which is generated as follows:
\begin{itemize}
\item We first sample a collection $\bu_1,\dots, \bu_d \sim \Exp(1)$ independently at random. We let $\bU = \diag(1/\bu_{1}^{1/p}, \dots, 1/\bu_d^{1/p})$ be the $d \times d$ matrix with entries $1/\bu_j^{1/p}$ across the diagonal, and let $\bV = \diag(\bU, \dots, \bU)$ denote the $(nd) \times (nd)$ matrix placing $n$ copies of the matrix $\bU$ in the diagonals.
\item Then, we sample $\bS^{(1)} = [\bS^{(1)}_1,\dots, \bS^{(1)}_n] \sim \calR$ and $\bS^{(2)} = [\bS^{(2)}_1,\dots,\bS^{(2)}_n] \sim \calO$. The final matrix is given by
\[ \bS = \left[ \begin{array}{c} \bS^{(1)} \\ \bS^{(2)} \end{array}\right] \cdot \bV. \]
Intuitively, the above operation consists of keeping two linear sketches: one of $\sum_{i=1}^n \bS^{(1)}_i (\lambda_i^{1/p} \by_i)$, and one of $\sum_{i=1}^n \bS^{(2)}_i (\lambda_i^{1/p} \by_i)$, where $\by_1,\dots, \by_n \in \R^d$ are given by letting, for each $j \in [d]$, 
\begin{align}
\by_{ij} &= \frac{1}{\bu_{j}^{1/p}} \cdot x_{ij}. \label{eq:transform-x-to-y}
\end{align}
\end{itemize}
We note that $\sum_{i=1}^n \lambda_i x_i = 0$ implies $\sum_{i=1}^n \lambda_i \by_i = 0$, as we simply re-scaled all coordinates of $x_{\cdot, j}$ by the same value. We now specify the description of the algorithm for outputting the triple $(\hat{\bj}, \hat{\balpha}, \hat{\bbeta})$, when given as input $\bS$ and $\sum_{i=1}^n \bS_i (\lambda_i^{1/p} x_i)$. We will proceed in the following way:
\begin{enumerate}
\item We will first apply the algorithm of Lemma~\ref{lem:p-stable-recovery} with the matrix $\bS^{(1)}$ and the vector $\sum_{i=1}^n \bS^{(1)}_i (\lambda_i^{1/p} \by_i)$. Notice that this algorithm outputs a sequence of $d$ numbers, $\balpha_1,\dots, \balpha_d \in \R_{\geq 0}$. We will set
\begin{align}
\hat{\bj} \eqdef \argmax_{j \in [d]} \balpha_j \qquad \text{and}\qquad \hat{\balpha} \eqdef \balpha_{\hat{\bj}} \cdot \bu_{\hat{\bj}}^{1/p}.  \label{eq:find-max}
\end{align}
\item Then, we will apply the algorithm of Lemma~\ref{lem:p-stable-recovery-opt} with the matrix $\bS^{(2)}$ and the vector $\sum_{i=1}^n \bS^{(2)} (\lambda_i^{1/p} \by_i)$, parameters $\lambda_1,\dots, \lambda_n$, and the input $\hat{\bj}$ and $\gamma = \hat{\balpha}$. The algorithm produces an output $\bbeta$, and we let
\[ \hat{\bbeta} \eqdef \bbeta \cdot \bu_{\bj}^{1/p}. \]
\end{enumerate}

We now consider the following event $\calE$, which is a function of the random variables $\bu_1,\dots, \bu_d$ and $x_1,\dots, x_n\in \R^d$. The goal is to show that (1) the event $\calE$ is satisfied for a random $\bu_1,\dots, \bu_d \sim \Exp(1)$ with high probability, and (2) if the event $\calE$ is satisfied, then $(\hat{\bj}, \hat{\balpha}, \hat{\bbeta})$ satisfies the conditions of Lemma~\ref{lem:main-sketch-lemma}.
\begin{definition}\label{def:event-e}
Consider a fixed setting of $u_1,\dots, u_d \in \R_{> 0}$, and let
\[ w_j \eqdef \frac{1}{u_j} \cdot \sum_{i=1}^n \lambda_i |x_{ij}|^p, \qquad j^* \eqdef \argmax_{j \in [d]} w_j, \qquad\text{and}\qquad j^{**} \eqdef \argmax_{j \in [d] \setminus \{ j^*\}} w_j. \] 
Finally, we also let
\[ \eta \eqdef \eps\delta 2^{-p} / 1000. \]
We say the event $\calE$ is satisfied for the setting of $u_1,\dots, u_d$ when the following conditions hold:
\begin{align}
w_{j^*} &\geq \eta \sum_{j=1}^d w_j \label{eq:E-condition-2} \\
w_{j^*} &\geq (1+\eta) w_{j^{**}}. \label{eq:E-condition-3}
\end{align}
\end{definition}

\begin{claim}\label{cl:claim-uno}
Consider a random setting of $\bu_1,\dots, \bu_d \sim \Exp(1)$, and let $\bj^* \in [d]$ be the random variable (which depends on $\bu_1,\dots, \bu_d$) which sets $\bj^*$ to the maximum $\bw_{j}$, as in Definition~\ref{def:event-e}. Then, $\bj^*$ is drawn from $\calD$.
\end{claim}

\begin{proof}
The proof is a simple identity of the exponential distribution. See Fact~5.3 of \cite{CJLW22} for the complete computation.
\end{proof}

\begin{claim}\label{cl:claim-uno-y-medio}
If we sample $\bu_1,\dots, \bu_{d} \sim \Exp(1)$, the event $\calE$ is satisfied with probability at least $1 - \eps\delta 2^{-p}/3$.
\end{claim}

\begin{proof}
This follows from Lemma~5.5 of \cite{CJLW22}. Namely, consider the vector $z \in \R^d$ where $z_j = \sum_{i=1}^n \lambda_i |x_{ij}|^p$. Lemma~5.5 implies $|z_{j^*}| / \bu_{j^*} \geq \eta \| z\|_1$ and $|z_{j^*}| / \bu_{j^*} \geq (1+\eta) |z_{j^{**}}| / \bu_{j^{**}}$ with probability at least $1 - 4\eta$, giving us (\ref{eq:E-condition-2}) and (\ref{eq:E-condition-3}).
\end{proof}

\begin{claim}\label{cl:claim-dos}
Let $u_1,\dots, u_d \in \R_{>0}$ be a fixed setting where event $\calE$ is satisfied. Then, let $y_1,\dots, y_n$ be given from $x_1,\dots, x_n$ by (\ref{eq:transform-x-to-y}), and run the algorithm of Lemma~\ref{lem:p-stable-recovery} while given access to $\bS^{(1)}$ and $\sum_{i=1}^n \bS^{(1)}_i (\lambda_i^{1/p} y_i)$, to output $\balpha_1,\dots, \balpha_d$. Then, with probability at least $1 - \delta_1$, we have
\begin{align*}
\argmax_{j \in [d]} \balpha_j &= j^* \qquad \text{and}\\
(1-\eps_2) \left( \sum_{i=1}^n \lambda_i |y_{ij^*}|^p\right)^{1/p}&\leq \balpha_{j^*} \leq (1+\eps_2) \left( \sum_{i=1}^n \lambda_i |y_{ij^*}|^p\right)^{1/p} 
\end{align*}
\end{claim}

\begin{proof}
We first note that applying Lemma~\ref{lem:p-stable-recovery}, as well as condition (\ref{eq:E-condition-2}), with probability at least $1 - \delta_1$, the algorithm outputs $\balpha_1, \dots, \balpha_d$ satisfying
\begin{align}
(1-\eps_1) w_j^{1/p} - \frac{\eps_1}{\eta^{1/p}} \cdot w_{j^*}^{1/p} \leq \balpha_j \leq (1+\eps_1) w_j^{1/p} + \frac{\eps_1}{\eta^{1/p}} \cdot w_{j^*}^{1/p}.\label{eq:guarantees-on-balpha}
\end{align}
In other words, we use (\ref{eq:E-condition-3}) to say
\begin{align*}
\left(1-\eps_1 - \frac{\eps_1}{\eta^{1/p}}\right) \cdot w_{j^*}^{1/p} \leq \balpha_{j^*} \qquad\text{and, for all $j \neq j^*$,}\quad \balpha_{j} \leq \balpha_{j^{**}} &\leq (1+\eps_1) w_{j^{**}}^{1/p} + \frac{\eps_1}{\eta^{1/p}} \cdot w_{j^{*}}^{1/p} \\
				&\leq \left(\frac{1}{1 + \eta /(2p)} + \eps_1 + \frac{\eps_1}{\eta^{1/p}} \right) w_{j^*}^{1/p}.
\end{align*}
Hence, the fact $\argmax_{j \in [d]} \balpha_j = j^*$ follows from the fact
\[ \frac{1}{1 + \eta/(2p)} + \eps_1 + \frac{\eps_1}{\eta^{1/p}} \leq 1  - \eps_1 - \frac{\eps_1}{\eta^{1/p}} \]
when $\eps_1 \leq c \eta^{1+1/p} / p$, for a small constant $c > 0$. The condition on $\balpha_{j^*}$ then follows from (\ref{eq:guarantees-on-balpha}) since $\eps_1 \left( 1 + 1/\eta^{1/p}\right) \leq \eps_2$.
\end{proof}

\begin{claim}\label{cl:claim-tres}
Let $u_1,\dots, u_d \in \R_{> 0}$ be a fixed setting where $\calE$ is satisfied. Then, let $y_1,\dots, y_n$ be given from $x_1,\dots, x_n$ according to (\ref{eq:transform-x-to-y}), and notice $\sum_{i=1}^n \lambda_i y_i = 0$. Run the algorithm of Lemma~\ref{lem:p-stable-recovery-opt} while given access to $\bS^{(2)}$, $\sum_{i=1}^n \bS^{(2)}_i (\lambda_i^{1/p} y_i)$, the parameters $\lambda_1,\dots, \lambda_n$,$j^*$, and a parameter $\gamma \in \R_{> 0}$ (which is set to $\hat{\balpha}$) satisfying
\[(1-\eps_2) \left(\sum_{i=1}^n \lambda_i |y_{ij^*}|^p \right)^{1/p} \leq \gamma \leq (1+\eps_2) \left( \sum_{i=1}^n \lambda_i |y_{ij^*}|^p\right)^{1/p}. \]
With probability at least $1 - \delta_2$, the algorithm outputs a parameter $\bbeta \in \R_{\geq 0}$ which satisfies
\[ (1-\eps) \min_{z \in \R} \left( \sum_{i=1}^n \lambda_i |y_{ij^*} - z|^p\right)^{1/p} \leq \bbeta \leq (1+\eps) \min_{z \in \R} \left( \sum_{i=1}^n \lambda_i |y_{ij^*} - z|^p \right)^{1/p} \]
\end{claim}

\begin{proof}
The proof proceeds similarly to above. Namely, the algorithm of Lemma~\ref{lem:p-stable-recovery-opt} outputs a parameter $\hat{\bbeta}$ which is a multiplicative $(1\pm \eps_2)$-approximation to $\min_{z \in \R} ( \sum_{i=1}^n |y_{ij^*} - z|^p )^{1/p}$.
with an additive error of $\eps_2 ( \sum_{j=1}^d \sum_{i=1}^n \lambda_i| y_{i,j}|^p)^{1/p}$. Hence, we will upper bound the additive error.
\begin{align*}
\sum_{j=1}^d \sum_{i=1}^n \lambda_i |y_{i, j}|^p &= \sum_{j=1}^d w_{j} \leq \frac{1}{\eta} \cdot w_{j^*},
\end{align*}
and the right-most expression is exactly $\frac{1}{ \eta} \sum_{i=1}^n \lambda_i |y_{ij^*}|^p$. Specifically, the additive error incurred is at most
\[ \frac{\eps_2}{\eta^{1/p}} \left(\sum_{i=1}^n \lambda_i |y_{ij^*}|^p \right)^{1/p} \leq \frac{2\eps_2}{\eta^{1/p}} \cdot \min_{z \in \R}\left(\sum_{i=1}^n \lambda_i |y_{ij^*} - z|^p \right)^{1/p}.\]
In particular, since $\eps_2 + 2\eps_2 / \eta^{1/p} \leq \eps$, we obtain the desired claim.
\end{proof}

We note that combining Claims~\ref{cl:claim-uno},~\ref{cl:claim-uno-y-medio},~\ref{cl:claim-dos}, and~\ref{cl:claim-tres} will give the desired lemma. With probability $1 - \eps\delta 2^{-p}/3$, event $\calE$ is satisfied by Claim~\ref{cl:claim-uno-y-medio}, and hence Claim~\ref{cl:claim-dos} imply that in the triple $(\hat{\bk}, \hat{\balpha}, \hat{\bbeta})$, we have $\hat{\bj} = j^*$ with probability at least $1 - \delta_1 = 1 - \eps\delta 2^{-p} / 3$. Furthermore, we always have, by definition of $\by$,
\begin{align*}
\bu_{j^*}^{1/p} \left( \sum_{i=1}^n \lambda_i |\by_{ij^*}|^p\right)^{1/p} &= \left(\sum_{i=1}^n \lambda_i|x_{ij^*}|^p \right)^{1/p} \\
\bu_{j^*}^{1/p} \min_{z \in \R} \left( \sum_{i=1}^n \lambda_i |\by_{ij^*} - z|^p \right)^{1/p} &= \min_{z \in \R} \left(\sum_{i=1}^n \lambda_i|x_{ij^*} - z|^p \right)^{1/p}.
\end{align*}
In particular, Claim~\ref{cl:claim-dos} implies $\hat{\balpha}$ is a $(1\pm\eps_2)$-multiplicative approximation. This establishes the approximation requirement of $\hat{\balpha}$, but it also allows us to utilize it in Claim~\ref{cl:claim-tres}, which implies the desired approximation on $\hat{\bbeta}$. Finally, the probability that $\hat{\bj} = j^*$ is at least the probability that event $\calE$ is satisfied, and the sketch of Lemma~\ref{lem:p-stable-recovery} succeeds, which both occur with probability at least $1 - \eps 2^{-p}$ by the setting of the failure probabilities on these events. By Claim~\ref{cl:claim-uno}, the variation distance of the random variable $\hat{\bj}$ from $\calD$ is at most $\eps 2^{-p}$.

%% file: derandomization.tex

\subsection{Storing the randomness}\label{sec:randomness}

The previous section described the $\ell_p^p$-median sketch assuming infinite precision on the random variables, as well as complete access to the randomness used. Here, we remark on how to: 1) utilize bounded precision on the entries of $\bS$ (to avoid storing infinite-precision real numbers), and 2) utilize Nisan's pseudorandom generator in order to maintain $\bS$ implicitly. This technique, first used in \cite{I06} for the $\ell_p$-sketch, is by now standard in the streaming literature. 

\paragraph{Bounded Precision} We assume that the input vectors $x_1,\dots, x_n \in \R^d$ and the weights $\lambda_1,\dots, \lambda_n$ have polynomially bounded entries (i.e., every number lies $[-\poly(nd), \poly(nd)]$ and is specified by $O(\log(nd))$ bits). The only infinite-precision numbers are the $p$-stable random variables and the exponential random variables, so we will discretize them so they require only $O(\log(nd))$ random bits to generate them. Similarly to Claim~3 of \cite{I06} our algorithm's output, as a function of the $p$-stable and exponential random variables, always considers polynomially bounded entries and has bounded derivatives in all but a $1/\poly(nd)$-fraction of the probability space. With probability $1 / \poly(nd)$, the $p$-stable random variables may become too large or exponentials too small. Furthermore, for any fixed setting of the $p$-stable random variables, a discontinuity in the algorithm's output occurs when the exponential random variables cause $\bk$ may suddenly switch, when the maximum of $\balpha_k$ in (\ref{eq:find-max}) is not unique. When we discretize the exponential distribution, the probability this occurs is at most $1/\poly(nd)$.

\paragraph{Using Nisan's Pseudorandom Generator} Instead of utilizing $s nd \cdot O(\log(nd))$ random bits to generate the matrix $\bS$, we show that we may utilize Nisan's pseudorandom generator. This technique is standard when the sketch is linear: we may re-order the elements of the stream so as to update the algorithm with a finite-state machine which reads the random bits in $O(\log(nd))$-sized chunks. The minor subtlety is that algorithm of Lemma~\ref{lem:main-sketch-lemma} requires access to $\bS$ in order to produce its output $(\bj, \balpha, \bbeta)$, and hence refers back to previous random bits. 

There is a simple fix: for each fixed setting of the input $x_1,\dots, x_n \in (\{0,1\}^{O(\log(nd))})^{nd}$ and weights $\lambda_1,\dots, \lambda_n \in (\{0,1\}^{O(\log(nd))})^{n}$, and $j' \in [d]$ and $y \in \R$ specified by $O(\log(nd))$ bits, there exists a low-space finite state machine which reads the randomness in $O(\log(nd))$ chunks and checks whether $\bj = j'$, and if so, checks that $\balpha$ satisfies the guarantees of Lemma~\ref{lem:main-sketch-lemma}, and accurately estimates $\bbeta$ if $z = y$. The finite state machine receives its input in the order of the coordinates; this way, it can temporarily store $\bu_j$ to compute $\balpha_j$ and find (\ref{eq:find-max}). Furthermore, while processing coordinate $j'$, it checks the cost of using $y$ as a center. Since Nisan's pseudorandom generator fools every such finite state machine, and we try at most $O(d) \cdot \poly(\log d , 1/\eps)$ possible coordinates $j'$ and centers $y$, a union bound will show that the output of the algorithm under true random bits and under Nisan's pseudorandom generator differs by at most $1/\poly(nd/\eps)$ in total variation distance. 

%% file: k-clustering.tex

\section{Streaming $(k,p)$-Clustering Costs in $\ell_p$}\label{sec:streaming}

We give the following application of our $\ell_p^p$-median sketch to estimating the $(k,p)$-clustering in $\ell_p$ spaces on a stream for $p \in [1,2]$. The space complexity of our streaming algorithm will be $\poly(\log(nd), k, 1/\eps)$, and will work in insertion-only streams.

\begin{theorem}\label{thm:streaming-alg}
Fix $n, d, k \in \N$, as well as $p \in [1,2]$ and $\eps \in (0,1)$. There exists an insertion-only streaming algorithm which processes a set of $n$ points $x_1,\dots, x_n \in \R^d$ utilizing $\poly(\log(nd), k, 1/\eps)$ space, and outputs a parameter $\boldeta \in \R$ which satisfies
\[ (1-\eps) \min_{\substack{C_1,\dots, C_k \\ \text{partition $[n]$}}} \sum_{\ell=1}^k \min_{c_{\ell} \in \R^d} \sum_{i \in C_{\ell}} \| x_i - c_{\ell} \|_p^p \leq \boldeta \leq (1+\eps) \min_{\substack{C_1,\dots, C_k \\ \text{partition $[n]$}}} \sum_{\ell=1}^k \min_{c_{\ell} \in \R^d} \sum_{i \in C_{\ell}} \| x_i - c_{\ell} \|_p^p \]
with probability at least $0.9$.
\end{theorem}

\subsection{Preliminaries: Coreset of \cite{HV20} and $\ell_p^p$-median sketch}

We will utilize the following result of \cite{HV20}.
\begin{definition}[Strong Coresets]\label{def:strong-coreset}
Let $x_1,\dots, x_n \in \R^d$ be a set of points and $p \in [1, 2]$. The set $S \subset [n]$ and $w \colon S \to \R_{\geq 0}$ is an $\eps$-\emph{strong coreset} for $(k, z)$-clustering in $\ell_p$ if for any set of $k$ points $c_1,\dots ,c_{\ell} \in \R^d$, 
\begin{align*}
\sum_{i=1}^n \min_{\ell \in[k]} \| x_i - c_{\ell} \|_p^z \approx_{1\pm \eps} \sum_{i \in S} w(i) \min_{\ell \in [k]} \| x_i - c_{\ell} \|_p^z.
\end{align*}
\end{definition}

\begin{corollary}[Corollary~5.18 of \cite{HV20}]\label{cor:strong-coresets}
For any $p \in [1,2]$, there exists a randomized algorithm that, given $O(1)$-approximate $\ell_p$-distances between $n$ points $x_1,\dots, x_n \in \R^d$, an integer $k \in \N$, $z \in [1,2]$, and $\eps, \delta \in (0, 1/2)$, outputs a (random) subset $\bS \subset [n]$ of $\poly(k\log(1/\delta)/\eps)$ points and a set of (random) weights $\bw \colon \bS \to \R_{\geq 0}$ which is a strong coreset for $(k, p)$-clustering in $\ell_p$ with probability at least $1-\delta$.\footnote{Corollary~5.18 is stated with $z \geq 2$; however, the proof works for any $z \in [1,2]$ as well.}
\end{corollary}

It is important that the algorithm of Corollary~\ref{cor:strong-coresets} only needs (approximate) access to distances between the $n$ points, as even storing a single point requires $\Omega(d)$ space. In particular, the $d$-dimensional representation of the coreset points is never stored. Even though \cite{HV20} does not explicitly mention this in Corollary~5.18, it can be verified by inspecting Algorithm~1. One additional fact is needed in Step~2 of Algorithm~1 of \cite{HV20}: a $O(1)$-approximate center set $C^*$ may be found using a subset of the dataset points when $z \geq 1$ is not too large.

\begin{fact}
Consider any $k \in \N$, $z \geq 1$ and $p \in [1, \infty)$, as well as points $x_1,\dots, x_n \in \R^d$. There exists a set of centers $c_1^*, \dots, c_k^* \in \{ x_1,\dots, x_n \}$ which give a $2^z$-approximation to the $(k,z)$-clustering in $\ell_p$.
\end{fact} 

\begin{proof}
Let $c_1,\dots, c_k \in \R^d$ denote the optimal centers for $(k,z)$-clustering $x_1,\dots, x_n$, and let $C_1,\dots,C_{k}$ denote the partition of $[n]$ from the center set $c_1,\dots, c_k$, i.e., for $\ell \in [k]$,
\[ C_{\ell} = \left\{ i \in [n] : \forall \ell ' \in [k], \| x_i - c_{\ell} \|_p \leq \| x_i - c_{\ell'} \|_p  \right\}, \]
with ties broken arbitrarily. For each $\ell \in [k]$, let $\bi_{\ell} \sim C_{\ell}$ be a uniformly random index from $C_{\ell}$, and let $\by_{\ell} = x_{\bi_{\ell}}$. Then, we show that the expected cost of using $\by_{\ell}$ as a center for cluster $C_{\ell}$ cannot increase the cost significantly,
\begin{align*}
\Ex_{\bi_{\ell} \sim C_{\ell}}\left[ \sum_{i \in C_{\ell}} \| x_i - \by_{\ell} \|_p^z \right] &\leq \Ex_{\bi_{\ell} \sim C_{\ell}}\left[ \sum_{i \in C_{\ell}} \left( \| x_i - c_{\ell} \|_p + \| \by_{\ell} - c_{\ell} \|_p \right)^z \right]  \qquad \text{(H\"{o}lder Inequality)}\\
&\leq 2^{z-1} \sum_{i \in C_{\ell}} \| x_i - c_{\ell} \|_p^z + \Ex_{\bi_{\ell} \sim C_{\ell}}\left[ |C_{\ell}| \cdot 2^{z-1} \| \by_{\ell} - c_{\ell} \|_p^z\right] = 2^{z} \sum_{i \in C_{\ell}} \| x_i - c_{\ell} \|_p^z.
\end{align*}
Hence, for each cluster $\ell \in [k]$, there exists a dataset point $c_{\ell}^*$ such that, utilizing $c_{\ell}^*$ instead of the optimal center $c_{\ell}$ incurs at most a factor of $2^z$ in the cost of that cluster, and therefore, the $(k,z)$-clustering cost can be at most $2^{z}$ times more than with the optimal centers $c_1,\dots, c_k$.
\end{proof}

We now re-state Theorem~\ref{thm:median-sketch} for the case that the weights $\lambda_1,\dots, \lambda_n$ are unknown at time of processing.
\begin{lemma}\label{lem:ell-p-median}
Fix $m , d \in \N$, as well as $p \in [1,2]$ and $\eps, \delta \in (0,1)$, let $s = \poly(\log d, 1/\eps, \log(1/\delta))$. There exists a distribution $\calS$ over $s \times (md)$ matrices and an algorithm which satisfies the following.
\begin{itemize}
\item For any set of points $x_1,\dots, x_m \in \R^d$, and any set of weights $\lambda_1,\dots, \lambda_n \in [0,1]$ with $\sum_{i=1}^n \lambda_i = 1$, we sample $\bS = [\bS_1,\dots, \bS_n] \sim \calS$ and give the algorithm as input the matrix $\bS$, the weights $\lambda_1,\dots, \lambda_m$, and the vectors
\[ \bS_j x_i \in \R^s \qquad\text{ for all $i, j \in [m]$}.\]
\item With probability $1-\delta$ over the draw of $\bS$, the algorithm outputs a parameter $\boldeta \in \R_{\geq 0}$ which satisfies
\[ \min_{y \in \R^d} \sum_{i=1}^n \lambda_i \| y - x_i\|_p^p \leq \boldeta \leq (1+\eps) \min_{y \in \R^d} \sum_{i=1}^n \lambda_i \|y - x_i \|_p^p.\]
\end{itemize}
\end{lemma}

\begin{proof}
The above is a simple consequence of Lemma~\ref{lem:main-sketch-lemma}. As in the proof of Theorem~\ref{thm:median-sketch}, the algorithm utilizes the algorithm of Lemma~\ref{lem:main-sketch-lemma} with the vector
\begin{align*}
\sum_{i=1}^m \lambda_i^{1/p} \bS_i x_i - \sum_{i=1}^m \sum_{j=1}^m \lambda_i \lambda_j^{1/p} \bS_j x_i &= \sum_{i=1}^m \lambda_i^{1/p} \bS_i \left( x_i - \sum_{j=1}^m \lambda_j x_j \right),
\end{align*}
which it can compute given access to $\bS_j x_i \in \R^d$ for all $i, j \in [m]$. The collection of vectors $x_i - \sum_{j=1}^m \lambda_j x_j$ is centered with respect to $\lambda_1,\dots, \lambda_m$, and is exactly the vector that Lemma~\ref{lem:main-sketch-lemma} expects to see.
\end{proof}

\subsection{Proof of Theorem~\ref{thm:streaming-alg}}

Let $M$ be an upper bound on the size of the coreset given from running the merge-and-reduce framework for coresets in the streaming model \cite{BS80} with $n$ points, where we use the coreset construction of Corollary~\ref{cor:strong-coresets}. Notice that $M = \poly(k \log(n)/\eps)$, where the extra $\poly(\log n)$-factors come from the merge-and-reduce framework. For each $m \in [M]$, we instantiate the distribution $\calS_{m}$ given by Lemma~\ref{lem:ell-p-median} with parameters $m, d$, $p$, $\eps$, and $\delta = k^{-M}/(100 k)$. We sample $\bS^{(1)},\dots, \bS^{(M)}$, where each $\bS^{(m)} = [\bS_1^{(m)},\dots, \bS_m^{(m)}] \sim \calS_m$. Note that for each $m \in [M]$, $\bS^{(m)}$ is an $s \times (md)$ matrix with $s = \poly(k \log(n)/\eps)$.

When a new vector $x_i \in \R^d$ comes in the stream, we generate the following:
\begin{itemize}
\item An $\ell_p$-sketch of $x_i \in \R^d$ for computing $O(1)$-approximate $\ell_p$-distances which succeeds with probability $1 - o(1/n^2)$ \cite{I06}. This sketch is used to compute distances between dataset points, as these are needed in order to compute the coreset of Corollary~\ref{cor:strong-coresets}.
\item For each $m \in [M]$ and $j \in [m]$, we store the vector $\bS_j^{(m)} x_i \in \R^s$. This prepares the $\ell_p^p$-median sketch of Lemma~\ref{lem:ell-p-median} to be evaluated on a cluster of size $m$. 
\end{itemize}

While processing the stream, we maintain a coreset $S \subset [n]$ with weights $w \colon S \to \R_{\geq 0}$ by using the merge-and-reduce framework of Bentley and Saxe \cite{BS80} for the streaming model. We maintain the $\ell_p$-sketch of $x_i$ and the $O(M^2)$ vectors $\bS_j^{(m)} x_i$ for all $m \in [M]$ and $j \in [m]$ whenever $i \in S$. If $i \notin S$, we no longer maintain the $\ell_p$-sketch or the vectors $\bS_j^{(m)} x_i$. This ensures that the total space complexity is $\poly(\log(nd), k, 1/\eps)$ because we store: 1) $O(M^3)$ vectors in $\R^s$, since there are $M$ points in the coreset, and each $i \in S$ stores the $O(M^2)$ vectors $\bS^{(m)}_j x_i \in \R^s$, 2) the weights $w(i)$ for each $i \in S$, and 3) the $\ell_p$-sketches for the points $x_i$.

When we finish processing the stream, we will evaluate the weighted cost of every clustering of $S$ coreset points. For each clustering, and for each cluster in the clustering, we will utilize the (weighted) $\ell_p^p$-median sketch in order to estimate its cost. First, notice that it suffices to compute the cost of all clustering of the coreset:
\begin{align*}
\min_{\substack{C_1,\dots, C_k \\ \text{partition $[n]$}}} \sum_{\ell=1}^k \min_{c_{\ell} \in \R^d} \sum_{i \in C_{\ell}} \| x_i - c_{\ell}\|_p^p &= \min_{c_1,\dots, c_k \in \R^d} \sum_{i=1}^n \min_{\ell \in [k]} \| x_i - c_{\ell} \|_p^p \\
	&\approx_{1\pm \eps} \min_{c_1,\dots, c_k \in \R^d} \sum_{i\in S} w(i) \min_{\ell \in [k]} \|x_i - c_{\ell} \|_p^p \quad\text{(def.~\ref{def:strong-coreset})}\\
	&= \min_{\substack{C_1,\dots, C_k \\ \text{partition $S$}}} \sum_{\ell=1}^k \min_{c_{\ell} \in \R^d} \sum_{i \in C_{\ell}} w(i) \| x_i - c_{\ell} \|_p^p.
\end{align*}
Consider a fixed partition $C_1,\dots, C_k$ of $S$. We evaluate the weighted cost of clustering $S$ with $C_1,\dots, C_k$ by evaluating the cost of each cluster $C_{\ell}$ for $\ell \in [k]$ and summing these costs up. Consider $C_{\ell} = \{ i_1,\dots, i_m \} \subset S$, and for $t \in [m]$, let
\[ \lambda_t = \frac{w(i_t)}{\sum_{j=1}^t w(i_j)}. \]
Recall that we stored the vectors $\bS^{(m)}_j x_{i_t} \in \R^s$, for all $j \in [m]$, so that we may utilize the algorithm of Lemma~\ref{lem:ell-p-median} to obtain an $(1\pm \eps)$-approximation to
\[ \boldeta \approx_{1\pm \eps} \min_{c_{\ell} \in \R^d} \sum_{t=1}^m \lambda_t \| x_{i_t} - c_{\ell} \|_p^p. \]
Then, $\boldeta \sum_{t=1}^m w(i_t)$ gives us a $(1\pm O(\eps))$-approximation to the cost of the cluster $C_{\ell}$. Since the $\ell_p^p$-median sketch succeeds with very high probability, we may union bound over all evaluations of at most $k$ cluster costs, for all $O(k^{M})$ possible clusterings of $S$. Hence, outputting the smallest cost gives us the desired approximation.

%% file: medoids-sketch.tex

\section{Sketching Medoid Costs}\label{app:medoid-cost}

Another notion of centrality of a set of points is the medoid: for a metric space $(X, d_X)$ and a set of $n$ points $x_1, \dots, x_n \in X$, the medoid of the set of points is 
\[ \argmin_{y \in \{ x_1,\dots, x_n \}} \sum_{i=1}^n d_{X}(x_i, y). \]
Analogously, we may define the $\ell_p^p$-medoid cost of a set of points $x_1,\dots, x_n \in \R^d$ as 
\[ \min_{y \in \{ x_1,\dots, x_n \}} \sum_{i=1}^n \| x_i - y\|_p^p.\]
Similarly to the case of $\ell_p^p$-medians, the hope is to leverage sketching algorithms for $\ell_p$-norms in order to sketch the $\ell_p^p$-medoid cost. However, as we now show, the $\ell_p^p$-median and $\ell_p^p$-medoid are very different from the sketching perspective. 

\begin{theorem}[Two-Pass Streaming Algorithm for $\ell_p^p$-Medoid]
Fix $n, d \in \N$, as well as $p \in [1,2]$ and $\eps \in (0, 1)$. There exists a two-pass, insertion-only streaming algorithm using space $\poly(\log(nd), 1/\eps)$ which processes a set of $n$ points $P = \{ x_1,\dots, x_n\} \in \R^d$, and outputs a parameter $\boldeta \in \R$ which satisfies
\begin{align*}
\min_{z \in P} \sum_{i=1}^n \| x_i - z\|_p^p \leq \boldeta \leq (1+\eps) \min_{z \in P} \sum_{i=1}^n \| x_i -z\|_p^p
\end{align*} 
with probability at least $1 - o(1)$.
\end{theorem}

The streaming algorithm will use the fact that there exists a linear sketch for $\ell_p$ norms. 
\begin{theorem}[Sketching $\ell_p$ \cite{I06}]\label{thm:ell_p-sketch}
For any $m \in \N$, $p \in [1,2]$, and $\eps, \delta \in (0,1)$, let $s = O(\log(1/\delta)/\eps^2)$. There exists a distribution $\calM$ over $s \times m$ matrices, and an algorithm which takes as input a vector in $\R^s$. For any $x \in \R^m$, with probability at least $1 - \delta$ over $\bS \sim \calM$, the algorithm on input $\bS x$, outputs $\boldeta \in \R_{\geq 0}$ satisfying
\begin{align*}
(1-\eps) \| x\|_p \leq \boldeta \leq (1+\eps) \|x\|_p. 
\end{align*}
\end{theorem}

\begin{proof}
We let $\calM$ be the distribution of Theorem~\ref{thm:ell_p-sketch} with accuracy parameter $\eps$ and failure probability $o(1/n)$. The sketch proceeds in the following way:
\begin{itemize}
\item We sample $\bS \sim \calM$. In the first pass, we consider the vector $x \in \R^{nd}$ which stacks the $n$ vectors $x_1, \dots, x_n \in \R^d$ and we maintain $\bS x$.
\item In the second pass, we process the vectors $x_1,\dots, x_n \in \R^d$. For $i = 1, \dots, n$, we let $y_i \in \R^{nd}$ be the vector which vertically stacks $x_i$ for $n$ repetitions. We query the algorithm on $\bS (x - y_i)$ to obtain an estimate $\boldeta_i$, and we maintain the minimum such $\boldeta_i$. 
\end{itemize}
We note that by the amplification of the success probability to $o(1/n)$, we may union bound over the $n$ queries to the streaming algorithm in the second pass.The bounds on the value of $\boldeta$ then follows from the fact that $\boldeta_i$ is an $(1\pm \eps)$-approximation to the $\ell_p^p$-medoid cost which utilizes $x_i$ as the medoid.
\end{proof}

%% file: medoids-1-pass-lb.tex

\begin{theorem}\label{thm:one-pass-lb}[Lower Bound for One-Pass $\ell_p^p$-Medoids] 
Fix $p \in [1,\infty)$ and $\eps \geq 10/n$. For large enough $n \in \N$, a one-pass, insertion-only streaming algorithm which processes a set $P$ of $n$ points in $\{0,1\}^{2n}$ and outputs a number $\boldeta \in \R_{\geq 0}$ satisfying
\[ \min_{z \in P} \sum_{x \in P} \| x- z\|_p^p \leq \boldeta \leq (2-\eps) \cdot \min_{z \in P} \sum_{x \in P} \| x - z\|_p^p \]
with probability at least $9/10$ must use $\Omega(\eps n)$ bits of space. 
\end{theorem}

In particular, this implies that any streaming algorithm which can approximate the medoid cost in one pass and $\polylog(nd)$ space achieves approximation no better than $2$. 

The proof follows a reduction to one-way communication complexity of indexing. 
\begin{definition}[Indexing Communication Problem] 
The Indexing communication problem is a two-party one-way communication game parametrized by $m \in \N$.
\begin{itemize}
\item Alice receives a bit-string $y \in \{0,1\}^m$, and Bob receives an index $i \in [m]$.
\item Alice and Bob may use public-randomness, and Alice must produce a single message to Bob so that Bob outputs $y_{i}$ with probability at least $9/10$.
\end{itemize}
\end{definition}

\begin{theorem}[One-way Communication of Indexing]
Any one-way communication protocol for the Indexing problem uses $\Omega(m)$ bits of communication. 
\end{theorem}

\begin{proof}[Proof of Theorem~\ref{thm:one-pass-lb}] 
We devise a protocol which utilizes a streaming algorithm for approximating the medoid in order to solve the indexing problem. Specifically, let $\alpha = \eps/10$ and suppose Alice receives a bit-string $y \in \{0,1\}^{\alpha n}$, where $\alpha n \in \N$. She considers the subset 
\[ P_A(y) = \{ e_{i} \in \{0,1\}^{2n} : y_i = 1 \}. \]
She process the stream of points $P_A(y)$ and communicates the contents of the memory and the number of points in $k =P_A(y)$. Bob, upon seeing $i \in [\alpha n]$ and $k \leq \alpha n$, needs to determine whether $e_{i} \in P_A(y)$ or $e_{i} \notin P_A(y)$. Bob considers the subset
\[ P_B(i,k) = \left\{ e_{i} + e_{j} : j \in \{ \alpha n + 1, \dots, \alpha n + n - k \}\right\},\]
so that $|P_B(i,k)| + |P_A(y)| = n$. Bob inserts the points $P_{B}(i, k)$ into the stream, which results in the subset of points $P = P_B(i,k) \cup P_A(y)$ of $n$ points in $\{0,1\}^{2n}$. The crucial point is that Bob may distinguish the two cases by considering the estimate $\boldeta$ of the medoid cost that the streaming algorithm produces. In particular,
\begin{itemize}
\item If $e_{i} \in P_A(y)$, then, by choosing a candidate center $z' = e_i$, we have
\[ \min_{z \in P} \sum_{x \in P} \| x - z\|_p^p \leq \sum_{x \in P} \|e_i - x\|_p^p \leq (k-1) \cdot 2 + (n-k) \leq  (1 + 2\alpha ) n.\]
In this case, the algorithm produces an estimate $\boldeta$ which is at most $(1+2\alpha) (2-\eps) n < 2n - 2$.
\item On the other hand, if $e_i \notin P_A(y)$, then every two points non-equal $x, x' \in P$ satisfy $\|x - x'\|_p^p \geq 2$, so that
\[ \min_{z \in P} \sum_{x \in P} \|x - z\|_p^p \geq 2\cdot (n-1).\]
\end{itemize}
Specifically, if the streaming algorithm was able to return a $(2-\eps)$-approximation, Bob could distinguish between the two cases. Since Bob communicates $O(\log_2(\eps n)) \ll \eps n$ bits to encode the number $k$, the space complexity of the streaming algorithm must be $\Omega(\eps n)$.
\end{proof}

%% file: main.bbl
\newcommand{\etalchar}[1]{$^{#1}$}
\begin{thebibliography}{KMNFT20}

\bibitem[ABIW09]{ABIW09}
Alexandr Andoni, Khanh~Do Ba, Piotr Indyk, and David Woodruff.
\newblock Efficient sketches for earth-mover distance, with applications.
\newblock In {\em Proceedings of the 50th Annual {IEEE} Symposium on
  Foundations of Computer Science ({FOCS}~'2009)}, 2009.

\bibitem[ACKS15]{ACKS15}
Pranjal Awasthi, Moses Charikar, Ravishankar Krishnaswamy, and Ali~Kemal Sinop.
\newblock The hardness of approximation of euclidean k-means.
\newblock In {\em 31st International Symposium on Computational Geometry (SoCG
  2015)}. Schloss Dagstuhl-Leibniz-Zentrum fuer Informatik, 2015.

\bibitem[ACNN11]{ACNN11}
Alexandr Andoni, Moses Charikar, Ofer Neiman, and Huy~L. Nguyen.
\newblock Near linear lower bound for dimension reduction in l1.
\newblock In {\em Proceedings of the 52nd Annual {IEEE} Symposium on
  Foundations of Computer Science ({FOCS}~'2011)}, 2011.

\bibitem[AHPV05]{AHV05}
Pankaj~K. Agarwal, Sariel Har-Peled, and Kasturi~R. Varadarajan.
\newblock Geometric approximation via coresets.
\newblock {\em Combinatorial and computational geometry}, 2005.

\bibitem[AIK08]{AIK08}
Alexandr Andoni, Piotr Indyk, and Robert Krauthgamer.
\newblock Earth mover distance over high-dimensional spaces.
\newblock In {\em Proceedings of the 19th {ACM-SIAM} Symposium on Discrete
  Algorithms ({SODA}~'2008)}, pages 343--352, 2008.

\bibitem[AKO10]{AKO10}
Alexandr Andoni, Robert Krauthgamer, and Krzysztof Onak.
\newblock Polylogarithmic approximation for edit distance and the asymmetric
  query complexity.
\newblock In {\em Proceedings of the 51st Annual {IEEE} Symposium on
  Foundations of Computer Science ({FOCS}~'2010)}, 2010.

\bibitem[AKO11]{AKO11}
Alexandr Andoni, Robert Krauthgamer, and Krzysztof Onak.
\newblock Streaming algorithms from precision sampling.
\newblock In {\em Proceedings of the 52nd Annual {IEEE} Symposium on
  Foundations of Computer Science ({FOCS}~'2011)}, 2011.

\bibitem[AKR15]{AKR15}
Alexandr Andoni, Robert Krauthgamer, and Ilya Razenshteyn.
\newblock Sketching and embedding are equivalent for norms.
\newblock In {\em Proceedings of the 47th {ACM} Symposium on the Theory of
  Computing ({STOC}~'2015)}, pages 479--488, 2015.
\newblock Available as arXiv:1411.2577.

\bibitem[AMS99]{AMS99}
Noga Alon, Yossi Matias, and Mario Szegedy.
\newblock The space complexity of approximating the frequency moments.
\newblock {\em Journal of Computer and System Sciences}, 58(1):137--147, 1999.

\bibitem[BBCA{\etalchar{+}}19]{BBCGS19}
Luca Becchetti, Marc Bury, Vincent Cohen-Addad, Fabrizio Grandoni, and Chris
  Schwiegelshohn.
\newblock Oblivious dimension reduction for $k$-means: Beyond subspaces and the
  johnson-lindenstrauss lemma.
\newblock In {\em Proceedings of the 51th {ACM} Symposium on the Theory of
  Computing ({STOC}~'2019)}, 2019.

\bibitem[BBCY17]{BBCKY17}
Jaroslaw Blasiok, Vladimir Braverman, Stephen~R. Chestnut, and Robert
  Krauthgamerand Lin~F. Yang.
\newblock Streaming symmetric norms via measure concentration.
\newblock In {\em Proceedings of the 50th {ACM} Symposium on the Theory of
  Computing ({STOC}~'2017)}, 2017.

\bibitem[BC05]{BC05}
Bo~Brinkman and Moses Charikar.
\newblock On the impossibility of dimension reduction in l1.
\newblock {\em Journal of the ACM}, 52(5):766--788, 2005.

\bibitem[Beh21]{B21}
Soheil Behnezhad.
\newblock Time-optimal sublinear algorithms for matching and vertex cover.
\newblock In {\em Proceedings of the 62nd Annual {IEEE} Symposium on
  Foundations of Computer Science ({FOCS}~'2021)}, 2021.

\bibitem[BFL16]{BFL16}
Vladimir Braverman, Dan Feldman, and Harry Lang.
\newblock New frameworks for offline and streaming coreset constructions.
\newblock {\em arXiv preprint arXiv:1612.00889}, 2016.

\bibitem[BFL{\etalchar{+}}17]{BFLSY17}
Vladimir Braverman, Gereon Frahling, Harry Lang, Christian Sohler, and Lin~F
  Yang.
\newblock Clustering high dimensional dynamic data streams.
\newblock In {\em Proceedings of the 34th International Conference on Machine
  Learning ({ICML}~'2017)}, 2017.

\bibitem[BGLT20]{BGLT20}
Guy Blanc, Neha Gupta, Jane Lange, and Li-Yang Tan.
\newblock Estimating decision tree learnability with polylogarithmic sample
  complexity.
\newblock In {\em Proceedings of Advances in Neural Information Processing
  Systems~33 ({NeurIPS}~'2020)}, 2020.

\bibitem[BHPI02]{BHI02}
Mihai Badoiu, Sariel Har-Peled, and Piotr Indyk.
\newblock Approximate clustering via core-sets.
\newblock In {\em Proceedings of the 34th {ACM} Symposium on the Theory of
  Computing ({STOC}~'2002)}, 2002.

\bibitem[BI14]{BI14}
Arturs Ba\v{c}kurs and Piotr Indyk.
\newblock Better embeddings for planar earth-mover distance over sparse sets.
\newblock In {\em Proceedings of the 41st International Colloquium on Automata,
  Languages and Programming ({ICALP}~'2014)}, 2014.

\bibitem[BIRW16]{BIRW16}
Arturs Backurs, Piotr Indyk, Ilya Razenshteyn, and David~P. Woodruff.
\newblock Nearly-optimal bounds for sparse recovery in generic norms, with
  applications to $k$-median sketching.
\newblock In {\em Proceedings of the 27th {ACM-SIAM} Symposium on Discrete
  Algorithms ({SODA}~'2016)}, pages 318--337, 2016.
\newblock Available as arXiv:1504.01076.

\bibitem[BS80]{BS80}
Jon~Louis Bentley and James~B Saxe.
\newblock Decomposable searching problems i. static-to-dynamic transformation.
\newblock {\em Journal of Algorithms}, 1(4):301--358, 1980.

\bibitem[BYJKS04]{BJKS04}
Ziv Bar-Yossef, T.S. Jayram, Ravi Kumar, and D.~Sivakumar.
\newblock An information statistics approach to data stream and communication
  complexity.
\newblock {\em Journal of Computer and System Sciences}, 68(4):702--732, 2004.

\bibitem[BZD10]{BZD10}
Christos Boutsidis, Anastasios Zouzias, and Petros Drineas.
\newblock Random projections for $k$-means clustering.
\newblock In {\em Proceedings of Advances in Neural Information Processing
  Systems~23 ({NeurIPS}~'2010)}, 2010.

\bibitem[CAK19]{CK19}
Vincent Cohen-Addad and CS~Karthik.
\newblock Inapproximability of clustering in lp metrics.
\newblock In {\em 2019 IEEE 60th Annual Symposium on Foundations of Computer
  Science (FOCS)}, pages 519--539. IEEE, 2019.

\bibitem[CASS21]{CSS21}
Vincent Cohen-Addad, David Saulpic, and Chris Schwiegelshohn.
\newblock A new coreset framework for clustering.
\newblock In {\em Proceedings of the 53rd {ACM} Symposium on the Theory of
  Computing ({STOC}~'2021)}, 2021.

\bibitem[CCF04]{CCF04}
Moses Charikar, Kevin Chen, and Martin Farach{-}Colton.
\newblock Finding frequent items in data streams.
\newblock {\em Theoretical Computer Science}, 312(1):3--15, 2004.

\bibitem[CEF{\etalchar{+}}05]{CEFMNRS05}
Artur Czumaj, Funda Eng\"{u}n, Lance Fortnow, Avner Magen, Ilan Newman, Ronitt
  Rubinfeld, and Christian Sohler.
\newblock Approximating the weight of the euclidean minimum spanning tree in
  sublinear time.
\newblock {\em {SIAM} Journal on Computing}, 2005.

\bibitem[CEM{\etalchar{+}}15]{CEMMP15}
Michael~B. Cohen, Sam Elder, Cameron Musco, Christopher Musco, and
  M\u{a}d\u{a}lina Persu.
\newblock Dimensionality reduction for $k$-means clustering and low rank
  approximation.
\newblock In {\em Proceedings of the 47th {ACM} Symposium on the Theory of
  Computing ({STOC}~'2015)}, 2015.

\bibitem[Che09]{C09}
Ke~Chen.
\newblock On coresets for k-median and k-means clustering in metric and
  euclidean spaces and their applications.
\newblock {\em {SIAM} Journal on Computing}, 39(3):923--947, 2009.

\bibitem[CJK{\etalchar{+}}22]{CJKVY22}
Artur Czumaj, Shaofeng H.-C. Jiang, Robert Krauthgamer, Pavel Vesel\'{y}, and
  Mingwei Yang.
\newblock Streaming facility location in high dimension via new geometric
  hashing.
\newblock {\em arXiv preprint arXiv:2204.02095}, 2022.

\bibitem[CJLW22]{CJLW22}
Xi~Chen, Rajesh Jayaram, Amit Levi, and Erik Waingarten.
\newblock New streaming algorithms for high dimensional emd and mst.
\newblock In {\em Proceedings of the 54th {ACM} Symposium on the Theory of
  Computing ({STOC}~'2022)}, 2022.

\bibitem[CKK20]{CKK20}
Yu~Chen, Sampath Kannan, and Sanjeev Khanna.
\newblock Sublinear algorithms and lower bounds for metric tsp cost estimation.
\newblock In {\em 47th International Colloquium on Automata, Languages, and
  Programming (ICALP 2020)}. Schloss Dagstuhl-Leibniz-Zentrum f{\"u}r
  Informatik, 2020.

\bibitem[CM05]{CM05}
Graham Cormode and Shan Muthukrishnan.
\newblock An improved data stream summary: the count-min sketch and its
  applications.
\newblock {\em Journal of Algorithms}, 2005.

\bibitem[CS09]{CS09}
Artur Czumaj and Chirstian Sohler.
\newblock Estimating the weight of metric minimum spanning trees in sublinear
  time.
\newblock {\em {SIAM} Journal on Computing}, 2009.

\bibitem[FL11]{FL11}
Dan Feldman and Michael Langberg.
\newblock A unified framework for approximating and clustering data.
\newblock In {\em Proceedings of the 43rd {ACM} Symposium on the Theory of
  Computing ({STOC}~'2011)}, 2011.

\bibitem[FSS13]{FSS13}
Dan Feldman, Melanie Schmidt, and Christian Sohler.
\newblock Turning big data into tiny data: constant-size coresets for k-means,
  pca and projective clustering.
\newblock In {\em Proceedings of the 24th {ACM-SIAM} Symposium on Discrete
  Algorithms ({SODA}~'2013)}, 2013.

\bibitem[HK20]{HK20}
Monika Henzinger and Sagar Kale.
\newblock Fully-dynamic coresets.
\newblock In {\em 28th Annual European Symposium on Algorithms (ESA 2020)}.
  Schloss Dagstuhl-Leibniz-Zentrum f{\"u}r Informatik, 2020.

\bibitem[HPM04]{HM04}
Sariel Har-Peled and Soham Mazumdar.
\newblock On coresets for $k$-means and $k$-median clustering.
\newblock In {\em Proceedings of the 36th {ACM} Symposium on the Theory of
  Computing ({STOC}~'2004)}, 2004.

\bibitem[HSYZ19]{HSFZ19}
Wei Hu, Zhao Song, Lin~F. Yang, and Peilin Zhong.
\newblock Nearly optimal dynamic $k$-means clustering for high-dimensional
  data.
\newblock {\em arXiv preprint arXiv:1802.00459}, 2019.

\bibitem[HV20]{HV20}
Lingxiao Huang and Nisheeth~K. Vishnoi.
\newblock Coresets for clustering in euclidean spaces: importance sampling is
  nearly optimal.
\newblock In {\em Proceedings of the 52nd {ACM} Symposium on the Theory of
  Computing ({STOC}~'2020)}, 2020.

\bibitem[Ind04]{I04b}
Piotr Indyk.
\newblock Algorithms for dynamic geometric problems over data streams.
\newblock In {\em Proceedings of the 36th {ACM} Symposium on the Theory of
  Computing ({STOC}~'2004)}, 2004.

\bibitem[Ind06]{I06}
Piotr Indyk.
\newblock Stable distributions, pseudorandom generators, embeddings, and data
  stream computation.
\newblock {\em Journal of the ACM}, 53(3):307--323, 2006.

\bibitem[IW05]{IW05}
Piotr Indyk and David Woodruff.
\newblock Optimal approximations of the frequency moments of data streams.
\newblock In {\em Proceedings of the 37th {ACM} Symposium on the Theory of
  Computing ({STOC}~'2005)}, pages 202--208, 2005.

\bibitem[JL84]{JL84}
William Johnson and Joram Lindenstrauss.
\newblock Extensions of {L}ipschitz mappings into a {H}ilbert space.
\newblock In {\em Conference in modern analysis and probability ({N}ew {H}aven,
  {C}onn., 1982)}, volume~26 of {\em Contemporary Mathematics}, pages 189--206.
  1984.

\bibitem[JW09]{JW09}
Thathachar~S. Jayram and David Woodruff.
\newblock The data stream complexity of cascaded norms.
\newblock In {\em Proceedings of the 50th Annual {IEEE} Symposium on
  Foundations of Computer Science ({FOCS}~'2009)}, 2009.

\bibitem[JW18]{JW18}
Rajesh Jayaram and David Woodruff.
\newblock Perfect lp sampling in a data stream.
\newblock In {\em Proceedings of the 59th Annual {IEEE} Symposium on
  Foundations of Computer Science ({FOCS}~'2018)}, 2018.

\bibitem[KBV20]{KBV20}
Weihao Kong, Emma Brunskill, and Gregory Valiant.
\newblock Sublinear optimal policy value estimation in contextual bandits.
\newblock In {\em Proceedings of the 23rd International Conference on
  Artificial Intelligence and Statistics ({AISTATS}~'2020)}, 2020.

\bibitem[KMNFT20]{KMNT20}
Michael Kapralov, Slobodan Mitrovi{\'c}, Ashkan Norouzi-Fard, and Jakab Tardos.
\newblock Space efficient approximation to maximum matching size from uniform
  edge samples.
\newblock In {\em Proceedings of the Fourteenth Annual ACM-SIAM Symposium on
  Discrete Algorithms}, pages 1753--1772. SIAM, 2020.

\bibitem[KNW10]{KNW10}
Daniel~M. Kane, Jelani Nelson, and David~P. Woodruff.
\newblock On the exact space complexity of sketching and streaming small norms.
\newblock In {\em Proceedings of the 21st {ACM-SIAM} Symposium on Discrete
  Algorithms ({SODA}~'2010)}, 2010.

\bibitem[KOR00]{KOR00}
Eyal Kushilevitz, Rafail Ostrovsky, and Yuval Rabani.
\newblock Efficient search for approximate nearest neighbor in high dimensional
  spaces.
\newblock {\em {SIAM} Journal on Computing}, 30(2):457--474, 2000.

\bibitem[KSS04]{KSS04}
Amit Kumar, Yogish Sabharwal, and Sandeep Sen.
\newblock A simple linear time $(1+\epsilon)$-approximation algorithm for
  $k$-means clustering in any dimension.
\newblock In {\em Proceedings of the 45th Annual {IEEE} Symposium on
  Foundations of Computer Science ({FOCS}~'2004)}, 2004.

\bibitem[KV18]{KV18}
Weihao Kong and Gregory Valiant.
\newblock Estimating learnability in the sublinear data regime.
\newblock In {\em Proceedings of Advances in Neural Information Processing
  Systems~31 ({NeurIPS}~'2018)}, pages 5455--5464, 2018.

\bibitem[LN04]{LN04}
James~R. Lee and Assaf Naor.
\newblock Embedding the diamond graph in lp and dimension reduction in l1.
\newblock {\em Geometric and Functional Analysis}, 14(4):745--747, 2004.

\bibitem[LS10]{LS10}
Michael Langberg and Leonard~J. Schulman.
\newblock Universal epsilon-approximators for integrals.
\newblock In {\em Proceedings of the 21st {ACM-SIAM} Symposium on Discrete
  Algorithms ({SODA}~'2010)}, 2010.

\bibitem[LSW17]{LMW17}
Euiwoong Lee, Melanie Schmidt, and John Wright.
\newblock Improved and simplified inapproximability for k-means.
\newblock {\em Information Processing Letters}, 120:40--43, 2017.

\bibitem[MMR19]{MMR19}
Konstantin Makarychev, Yuri Makarychev, and Ilya Razenshteyn.
\newblock Performance of johnson-lindenstrauss transform for $k$-means and
  $k$-medians clustering.
\newblock In {\em Proceedings of the 51th {ACM} Symposium on the Theory of
  Computing ({STOC}~'2019)}, 2019.

\bibitem[NO08]{NO08}
Huy~N Nguyen and Krzysztof Onak.
\newblock Constant-time approximation algorithms via local improvements.
\newblock In {\em 2008 49th Annual IEEE Symposium on Foundations of Computer
  Science}, pages 327--336. IEEE, 2008.

\bibitem[ORRR12]{ORRR12}
Krzysztof Onak, Dana Ron, Michal Rosen, and Ronitt Rubinfeld.
\newblock A near-optimal sublinear-time algorithm for approximating the minimum
  vertex cover size.
\newblock In {\em Proceedings of the twenty-third annual ACM-SIAM symposium on
  Discrete Algorithms}, pages 1123--1131. SIAM, 2012.

\bibitem[Owe13]{O13}
Art~B. Owen.
\newblock {\em Monte Carlo theory, methods, and examples}.
\newblock 2013.

\bibitem[PR07]{PR07}
Michal Parnas and Dana Ron.
\newblock Approximating the minimum vertex cover in sublinear time and a
  connection to distributed algorithms.
\newblock {\em Theoretical Computer Science}, 381(1-3):183--196, 2007.

\bibitem[SS02]{SS02}
Michael Saks and Xiaodong Sun.
\newblock Space lower bounds for distance approximation in the data stream
  model.
\newblock In {\em Proceedings of the 34th {ACM} Symposium on the Theory of
  Computing ({STOC}~'2002)}, 2002.

\bibitem[SW18]{SW18}
Christian Sohler and David Woodruff.
\newblock Strong coresets for $k$-median and subspace approximation.
\newblock In {\em Proceedings of the 59th Annual {IEEE} Symposium on
  Foundations of Computer Science ({FOCS}~'2018)}, 2018.

\bibitem[YYI09]{YYI09}
Yuichi Yoshida, Masaki Yamamoto, and Hiro Ito.
\newblock An improved constant-time approximation algorithm for maximum\~{}
  matchings.
\newblock In {\em Proceedings of the forty-first annual ACM symposium on Theory
  of computing}, pages 225--234, 2009.

\end{thebibliography}
